\documentclass[journal]{IEEEtran}
\usepackage{booktabs}
 \pdfoutput=1

\ifCLASSINFOpdf
\else
\fi

\usepackage[cmex10]{amsmath}
\usepackage{tikz,pgfplots}
\usepackage{graphicx}
\usepackage{placeins}
\usepackage{mfirstuc}
\usepackage{mathtools}
\usepackage{wrapfig}
\usepackage{pdfpages}
\usepackage{diagbox}
\usepackage{makecell}
\usepackage{amssymb}
\usepackage{bbm}
\usepackage[
ruled,
vlined,
boxed, 
linesnumbered, 
commentsnumbered]{algorithm2e}
\usepackage{siunitx}
\usepackage{amsthm}
\newtheorem{theorem}{Theorem}

\theoremstyle{definition}
\usepackage{enumitem}
\tikzset{
  pics/cross/.style={
    code={
      \draw[line width=1.5pt, red!90, line cap=round] (-0.1,-0.1) -- (0.1,0.1);
      \draw[line width=1.5pt, red!90, line cap=round] (0.1,-0.1) -- (-0.1,0.1);
    }
  },
  pics/crossblue/.style={
    code={
      \draw[line width=1.5pt, blue!70, line cap=round] (-0.1,-0.1) -- (0.1,0.1);
      \draw[line width=1.5pt, blue!70, line cap=round] (0.1,-0.1) -- (-0.1,0.1);
    }
  }
}

\newtheoremstyle{definitionStyle}%
  {0pt}%
  {0pt}%
  {\normalfont}%
  {}%
  {\bfseries}%
  {.}%
  {.5em}%
  {}%
  
\theoremstyle{definitionStyle}
\newtheorem{definition}{Definition}

\newtheoremstyle{definitionStyle}%
  {0pt}%
  {0pt}%
  {\normalfont}%
  {}%
  {\bfseries}%
  {.}%
  {.5em}%
  {}%
  
\theoremstyle{definitionStyle}

\usetikzlibrary{backgrounds,shapes,arrows,positioning}
\usetikzlibrary{shapes,snakes}
\usetikzlibrary{fit,calc} %

\usepackage[
    colorlinks=true,
    bookmarks=false,
    citecolor=blue,
    urlcolor=blue,
    linkcolor=blue
]{hyperref}

\usepackage{cite}
\usepackage{textcomp}
\usepackage{acronym}
\let\mathcal\undefined
\DeclareMathAlphabet{\mathcal}{OMS}{cmsy}{m}{n}

\newcommand{\CW}{\mathcal{C}}

\newcommand{\dmin}{d_{\mathrm{min}}}
\newcommand{\dGPC}{d_{\mathrm{GPC}}}
\newcommand{\ddesign}{d_{\mathrm{des}}}

\newcommand{\que}{\mathord{?}}

\newcommand{\ZO}{\{0, 1\}}
\newcommand{\ZQO}{\{0, \que, 1\}}

\usepackage{bm}

\newcommand{\BDD}{\text{BDD}}

\newcommand{\weight}[1]{\text{wt}\left(#1\right)}

\newcommand{\dnE}[1]{d_{\sim{\text{E}(#1)}}}
\newcommand{\dH}{d}

\newcommand{\E}{\text{E}}

\newcommand{\Eb}{E_{\text{b}}}
\newcommand{\No}{N_{\text{0}}}
\newcommand{\deltaC}{\delta_{\text{c}}}
\newcommand{\epsilonC}{\epsilon_{\text{c}}}

\newcommand{\ta}{\mathsf{T}_{\text{a}}}
\newcommand{\te}{\mathsf{T}}
\newcommand{\tee}{\mathsf{T}_{\text{e}}}

\newcommand{\yone}{\bm{y}^{(\text{1})}}
\newcommand{\ytwo}{\bm{y}^{(\text{2})}}
\newcommand{\pone}{\bm{p}^{(\text{1})}}
\newcommand{\ptwo}{\bm{p}^{(\text{2})}}

\newcommand{\wi}{\bm{w}^{(i)}}

\usepackage{stmaryrd}

\newcommand{\phicv}{\phi_{\mathsf{c}\shortrightarrow\mathsf{v}}}
\newcommand{\phicc}{\phi_{\mathsf{c}\shortrightarrow\mathsf{c}}}

\acrodef{AD}{anchor decoding}
\acrodef{AGC}[AGC]{automatic gain control}
\acrodef{API}[API]{application program interface}
\acrodef{AWGN}[AWGN]{additive white Gaussian noise}
\acrodef{AVX}[AVX]{advanced vector extensions}

\acrodef{BCH}{Bose--Chaudhuri--Hocquenghem}
\acrodef{BDD}{bounded distance decoding}
\acrodef{BEC}[BEC]{binary erasure channel}
\acrodef{BEE-PC}{binary message passing based on \ac{EaE} decoding for \acp{PC}}
\acrodef{BER}[BER]{bit error rate}
\acrodef{BDMC}[BDMC]{binary \acs{DMC}}
\acrodef{BI-AWGN}{binary input additive white Gaussian noise}
\acrodef{BM}{Berlekamp--Massey}
\acrodef{BP}[BP]{belief propagation}
\acrodef{BPSK}{binary phase shift keying} 
\acrodef{BSC}[BSC]{binary symmetric channel}

\acrodef{CEL}[CEL]{Communications Engineering Lab}
\acrodef{CPU}[CPU]{central processing unit}
\acrodef{CN}{check node}
\acrodef{DFT}[DFT]{discrete Fourier transform}
\acrodef{DMC}[DMC]{discrete memoryless channel}
\acrodef{DRS}{dynamic reliability score}
\acrodef{DRSD}{dynamic reliability score decoder}
\acrodef{RDRSD}[rDRSD]{refined dynamic reliability score decoder}
\acrodef{dSNR}[dSNR]{design \ac{SNR}}
\acrodef{DSP}[DSP]{digital signal processing}
\acrodef{DTP}{decoding transition probability}
\acrodef{EaE}{error-and-erasure}
\acrodef{EaED}{\ac{EaE} decoder}
\acrodef{EEP}{error-evaluator polynomial}
\acrodef{ELP}{error-locator polynomial}
\acrodef{EMP}{extrinsic message passing}
\acrodef{FEC}[FEC]{forward error correction}
\acrodef{FER}[FER]{frame error rate}
\acrodef{FFT}[FFT]{fast Fourier transform}
\acrodef{GLDPC}[GLDPC]{generalized \ac{LDPC}}
\acrodef{GPP}[GPP]{general purpose processor}
\acrodef{GPC}{generalized product code}
\acrodef{GMD}{generalized minimal distance}
\acrodef{GMDD}{generalized minimal distance decoder}
\acrodef{HD}{hard decision}
\acrodef{HDD}{hard decision decoding}
\acrodef{HPC}{half product code}
\acrodef{HIHO}{hard input hard output}
\acrodef{HRB}{highly reliable bit}
\acrodef{iBDD}{iterative bounded distance decoding}
\acrodef{IDFT}[IDFT]{inverse discrete Fourier transform}
\acrodef{iEaED}{iterative error-and-erasure decoding}

\acrodef{KIT}[KIT]{Karlsruhe Institute of Technology}
\acrodef{LDPC}[LDPC]{low-density parity-check}
\acrodef{LR}[LR]{likelihood ratio}
\acrodef{LLR}[LLR]{log-likelihood ratio}
\acrodef{LTE}[LTE]{Long Term Evolution}
\acrodef{MAP}[MAP]{maximum a posteriori}
\acrodef{MC}{miscorrection}
\acrodef{MD}{miscorrection detection}
\acrodef{MDS}{maximum distance separable}
\acrodef{ML}[ML]{maximum likelihood}
\acrodef{MSP}{minimal stall pattern}
\acrodef{NCG}{net coding gain}
\acrodef{OSD}{ordered statistics decoding}
\acrodef{OTN}{optical transport network}
\acrodef{PC}{product code}
\acrodef{PCM}{parity check matrix}
\acrodefplural{PCM}[PCMs]{parity check matrices}
\acrodef{PDF}{probability density function}

\acrodef{RS}{Reed--Solomon}
\acrodef{RV}[RV]{random variable}

\acrodef{SABM}{soft-aided bit marking}
\acrodef{SA-HDD}{soft-aided \ac{HDD}}
\acrodef{SABM-SR}{SABM with scaled reliabilities}
\acrodef{SC}{spatially-coupled}
\acrodef{SDD}{soft decision decoding}
\acrodef{SISO}{soft input soft output}
\acrodef{SNR}[SNR]{signal-to-noise ratio}
\acrodef{SPC}[SPC]{single parity check}
\acrodef{TPD}{turbo product decoding}
\acrodef{UB}{union bound}
\acrodef{VN}{variable node}

\definecolor{kit-green100}{rgb}{0,.59,.51}
\definecolor{kit-green70}{rgb}{.3,.71,.65}
\definecolor{kit-green50}{rgb}{.50,.79,.75}
\definecolor{kit-green30}{rgb}{.69,.87,.85}
\definecolor{kit-green15}{rgb}{.85,.93,.93}
\definecolor{KITgreen}{rgb}{0,.59,.51}

\definecolor{KITpalegreen}{RGB}{130,190,60}
\definecolor{KITcyanblue}{RGB}{80,170,230}
\definecolor{KITorange}{rgb}{.87,.60,.10}
\definecolor{tab104}{RGB}{152,78,163}
\definecolor{colorDRSD}{RGB}{255,128,0}
\definecolor{colorideal}{RGB}{224,243,248}
\definecolor{colorDRSDdeter}{RGB}{168,129,188}
\definecolor{colorDRSDplus20}{RGB}{255,102,178}
\definecolor{colorBEE-PC}{RGB}{96,96,96}
\definecolor{colorAD}{RGB}{0,102,51}
\definecolor{colorDRSD3bit}{RGB}{252,174,145}
\definecolor{colorDRSD4bit}{RGB}{251,106,74}
\definecolor{colorDRSD6bit}{RGB}{165,15,21}

\definecolor{cR1}{rgb}{0,.59,.51}
\definecolor{cR2}{RGB}{162,34,35}
\definecolor{cR12}{RGB}{217,115,14}

\definecolor{ber1}{RGB}{22,74,132}
\definecolor{ber2}{RGB}{52,120,178}
\definecolor{ber3}{RGB}{86,142,195}
\definecolor{ber4}{RGB}{110,180,210}

\begin{document}

\title{Low-Complexity Soft-Aided Error-and-Erasure Decoding for Generalized Product Codes}

\author{Sisi Miao,~\IEEEmembership{Student~Member,~IEEE} and Laurent Schmalen,~\IEEEmembership{Fellow,~IEEE}
\thanks{This work has received funding from the European Research Council (ERC)
under the European Union’s Horizon 2020 research and innovation programme
(grant agreement No. 101001899) and the German Federal Ministry of Research, Technology and Space (BMFTR) under grant agreement 16KIS2081 (PONGO).}
\thanks{All authors are with Karlsruhe Institute of Technology (KIT), Commu\-nications Engineering Lab (CEL), Hertzstr. 16, 76187 Karlsruhe, Germany. E-mail: \{\texttt{laurent.schmalen@kit.edu\}}}}

\markboth{Journal of \LaTeX\ Class Files,~Vol.~14, No.~8, August~2015}%
{Shell \MakeLowercase{\textit{et al.}}: Bare Demo of IEEEtran.cls for IEEE Journals}

\maketitle

\IEEEpeerreviewmaketitle

 \begin{abstract}
    Generalized product codes (GPCs) combine excellent high-rate performance with low-complexity hardware implementations.
    We propose the refined dynamic reliability score decoder (\acs{RDRSD}), a hard-message-passing iterative error-and-erasure decoder that uses dynamic reliability scores.
    Its syndrome-domain implementation has complexity comparable to hard-decision \ac{iBDD}.
    Across various GPCs and decoding configurations, \acs{RDRSD} provides different complexity--performance trade-offs and achieves approximately \SI{1}{dB} coding gain over \ac{iBDD}.
    For GPCs with component codes of small error-correcting capability, we also analyze the error floor and propose a soft-aided post-processing step that significantly lowers it.
\end{abstract}

\acresetall

\IEEEpeerreviewmaketitle

\section{Introduction}

\Acp{GPC}, or product-like codes, are widely used in high-throughput systems, particularly \acp{OTN}~\cite{ITU_G9751_2018,OIF_400ZR_2020}, due to their excellent high-rate performance and efficient hardware implementation. Their regular structure enables massive parallelization, and their algebraic component codes support syndrome-domain \ac{iBDD} with up to two orders of magnitude less internal data flow than \ac{LDPC} codes~\cite{staircaseCode}.

To improve GPC decoding and meet stringent optical-communication \ac{FEC} requirements~\cite{graell20forward}, \ac{SDD} algorithms have been introduced.
A well-known example is Chase--Pyndiah decoding~\cite{pyndiahNearoptimumDecodingProduct1998} and its low-complexity and soft-output variants~\cite{AlDweik2009hybrid,StraßhoferGMI,miao2026improved}.
Other approaches, such as iterative \ac{OSD}~\cite{SheniOSD} and \ac{BP} decoding~\cite{SheniBP}, achieve gains of up to \SI{0.2}{dB} over conventional Chase--Pyndiah decoding at moderately higher complexity.
However, \ac{SDD} is challenging to implement beyond \SI{1}{Tb/s}, since soft-message passing typically requires internal data flow several times the net data rate. The demand for higher rates and lower energy consumption has therefore motivated soft-aided hard-decision decoding.
These algorithms approach \ac{iBDD} complexity by incorporating limited soft information~\cite{sheikh2018iterative,liga2019novel,SheniOSD,ZhaoThreshold,ZhuStaircase,sheikh2021novel,11313545}.
Although they provide significant waterfall gains, their implementation complexity and error-floor behavior are often not fully characterized.

In previous work~\cite{miao2022JLT,Miao22ECOC,Rapp24ECOC}, we proposed the \ac{DRSD}, a low-complexity, high-performance soft-aided decoder for GPCs. Here, we introduce the refined DRSD (\acs{RDRSD}), which reduces complexity and memory while preserving or improving performance. It avoids the sorting required by \ac{DRSD} and reduces message passing for soft-information updates. The \acs{RDRSD} uses an \ac{iBDD}-like syndrome-domain architecture and gains about \SI{1}{dB} over \ac{iBDD} across the considered GPCs. We also analyze its PC error floor, derive a union bound for miscorrection-free \ac{EaE} decoding, and propose soft-aided post-processing to lower the floor. Secs.~\ref{sec:preliminary}--\ref{sec:DRSD} introduce the preliminaries and decoders; Secs.~\ref{sec:waterfallResults}--\ref{sec:spr} present waterfall results, error-floor analysis, and post-processing, followed by the conclusion.

\textit{Notation}: We use boldface letters to denote vectors, e.g., $\bm{y}$, and $y_i$ denotes the $i$th component of $\bm{y}$. With ${\mathbb{F}_2=\{0,1\}}$, we introduce the Hamming weight ${\weight{\bm{y}}=|\{i:y_i\neq 0\}|}$ for ${\bm{y}\in\mathbb{F}_2^n}$, the Hamming distance $\mathrm{d}({\bm{y}_1},{\bm{y}_2})=\weight{\bm{y}_1-\bm{y}_2}$, and the notation
${\dnE{\bm{y}'}(\bm{y}',\bm{c}):=|\{i\in \{1,2,\ldots,n\}:c_i\neq y'_i,y'_i\neq \que\}|}$ for ${\bm{c}\in \mathbb{F}_2^n}$ and ${\bm{y}'\in\{0,1,\que\}}$ where ``$\que$'' is an erasure. The operator $\E(\bm{y}')$ gives the number of erasures in $\bm{y}'$.
We define a function $\mu(x) = 1-2x$ and $\psi(x) = (1-\text{sign}(x))/2$. Both $\mu(x)$ and $\psi(x)$ apply to matrices element-wise. Finally, $Q(x)\!=\!\tfrac{1}{\sqrt{2\pi}}\int_x^{\infty}\text{e}^{-\frac{u^2}{2}}\mathrm{d}u$ is the Gaussian tail distribution function.

\section{Preliminaries}\label{sec:preliminary}
\subsection{Channel Models}
We consider three communication channels: the \ac{BSC}, the \ac{EaE} channel, and the \ac{BI-AWGN} channel.

A \ac{BSC} is defined by the crossover probability ${\delta\!=\!P(Y\!\neq \!X)}$.

For the \ac{EaE} channel, ${X\in\{0,1\}}$ and ${Y\in\{0,1,\que\}}$, where ``$\que$'' is an erasure. The transition probabilities are ${P(Y=1-x|X=x)=\deltaC}$, ${P(Y=\que|X=x)=\epsilonC}$, and ${P(Y=x|X=x)=1-\deltaC-\epsilonC}.$

For the \ac{BI-AWGN} channel with \ac{BPSK}, we map bit ${x\mapsto \mu(x)=:\tilde{x}}$, and the received symbol is ${\tilde{y}_i = \tilde{x}_i + n_i\in \mathbb{R}}$, where ${n_i\sim\mathcal{N}(0,\sigma^2)}$ with ${\sigma^2 = (2R\cdot \Eb/\No)^{-1}}$ and $R$ is the rate of the GPC.

Quantizing the \ac{BI-AWGN} output yields the \ac{BSC} and \ac{EaE} channels. For an erasure threshold \({\te\geq0}\), declare $y=\que$ if \({\tilde{y}\in[-\te,\te]}\), and otherwise set \(y=\psi(\tilde{y})\in\{0,1\}\). The resulting \ac{EaE} parameters are
\begin{equation}\label{eq:EaEprob}
    \deltaC =  Q\left(\tfrac{\te+1}{\sigma}\right), \quad \epsilonC = 1 - Q\left(\tfrac{\te-1}{\sigma}\right) - \deltaC.
\end{equation}
For \(\te=0\), this reduces to a \ac{BSC} with \(\delta = Q\left(\tfrac{1}{\sigma}\right)\).

\subsection{\Ac{BCH} Codes}
We use binary \ac{BCH} component codes because of their cyclic structure, large minimum distances, and efficient algebraic \ac{BDD}.

Let $\mathcal{C}_{\text{BCH}}[n, k, \dmin]$ denote a primitive, narrow-sense, binary \ac{BCH} code with block length $n$, dimension $k$, and minimum Hamming distance $\dmin$, yielding code rate $R_{\text{c}} = k/n$. Let $t$ denote the number of correctable errors and $\ddesign$ the designed distance, which is a lower bound on $\dmin$.\footnote{Throughout this work, all considered primitive BCH codes saturate this bound, i.e., $\dmin=\ddesign$.} For primitive BCH codes, $n = 2^b - 1$, $k \geq n - bt$, and $\ddesign = 2t + 1$, where $\mathbb{F}_{2^b}$ is the finite field over which the associated \ac{RS} supercode is defined. For brevity, we sometimes describe the parameters of a BCH code with $[n,t]$.

High-rate GPCs typically use BCH component codes with small minimum distance. Adding a parity bit makes all codewords even-weight and gives an extended BCH (eBCH) code.\footnote{The even-weight BCH subcode performs almost identically in our simulations; hence, we consider only eBCH codes.} Extending $[n,k,\dmin]$ yields $\mathcal{C}_{\text{eBCH}}[n+1,k,\dmin+1]$: the error-correcting capability $t$ is unchanged, while the erasure-correcting capability increases by $1$.

\subsection{Bounded Distance Decoding (BDD) for BCH Codes}\label{sec:BDDforBCH}

The \ac{BDD} rule for a binary vector $\bm{y}\in \ZO^n$ is
\begin{equation*}
    \BDD(\bm{y})=\begin{cases}
    \bm{c}& \exists \bm{c}\in \CW \text{ such that }\mathrm{d}(\bm{y},\bm{c}) \leq t\\
    \bm{y}&\text{otherwise}.
    \end{cases}
\end{equation*}
\Ac{BDD} corrects $\bm{y}$ to $\bm{c}$ if $\mathrm{d}(\bm{y},\bm{c})\leq t$ and otherwise declares failure. A \textit{miscorrection} occurs if it returns $\bm{c}\neq\bm{x}$, where $\bm{x}$ was transmitted, and can significantly degrade hard-decision GPC decoding (Sec.~\ref{sec:EaED}).

The \ac{BDD} algorithm for BCH codes is typically implemented by syndrome-based algebraic decoding: after computing the syndrome, the decoder obtains the error-locator polynomial, e.g., using the \ac{BM} algorithm, whose roots are then found by a Chien search to identify the error locations~\cite{roth_introduction_2006}.

For high-rate \ac{BCH} codes, the realization of \ac{BDD} can be significantly simplified, requiring only a small constant number of additions, multiplications, and lookup-table operations. This is achieved by parameterizing the solution of the error-locator polynomial in terms of the syndrome, allowing it to be directly obtained by substituting the syndrome values. Consequently, the Chien search can be replaced by a lookup-table operation, since the error-locator polynomial has low degree for small \(t\), as described in~\cite[Appendix~C]{SukmadjiRSBCH}. For example, for a $t=2$ code, a total of $1$ addition, $5$ multiplications, $1$ inversion, $1$ division, and $1$ lookup-table operation are needed to execute a \ac{BDD} step in the worst case. For a $t=3$ code, $8$ additions, $10$ multiplications, $2$ divisions, and $2$ lookup-table operations are needed.

\subsection{Generalized Product Codes}

We introduce GPCs through generalized Tanner graphs~\cite{Tanner81recursive}, using product and staircase codes as block and \ac{SC} examples.

For a block $[N,K,\dGPC]$ \ac{GPC}, its generalized Tanner graph is a bipartite graph consisting of a set $\mathcal{N}=\{\mathsf{v}_1,\mathsf{v}_2,\ldots,\mathsf{v}_N\}$ of \acp{VN}, with $|\mathcal{N}|=N$, and a set ${\mathcal{M}=\{\mathsf{c}_1,\mathsf{c}_2,\ldots,\mathsf{c}_M\}}$ of \acp{CN}, with $|\mathcal{M}|=M$. Each \ac{VN} represents a data bit, while each \ac{CN} represents a component code, which in this work is a binary BCH code or its extended code. An edge connects a VN and a CN if the corresponding data bit participates in the associated BCH codeword. The resulting GPC block length $N$, code dimension $K$, and minimum distance $\dGPC$ are determined by the structure of the generalized Tanner graph and the parameters of the component codes.

\ac{SC} GPCs append new component codewords continuously during transmission, yielding a generalized Tanner graph with infinitely many nodes. Since practical \ac{SC} GPCs are often periodic, their structure can be specified by describing one period of the generalized Tanner graph.

In this work, we focus on GPCs in which all \acp{VN} have degree two and all \acp{CN} employ the same component code. These restrictions enable simple hardware implementation and include many practical GPC constructions such as PCs and staircase codes. During iterative decoding, the degree-two \acp{VN} simply exchange messages between their two incident \acp{CN}. At the \acp{CN}, BCH decoding algorithms are executed to update the data buffer and, when available, the associated soft information.

Before introducing the example codes, we define two mappings $\phicv$ and $\phicc$ to facilitate the description of the iterative decoding algorithms. The mapping $\phicv:(j,\kappa)\mapsto i$ maps the $\kappa$-th incident VN of CN $\mathsf{c}_j$ to the VN $\mathsf{v}_i$. The mapping $\phicc:(j,\kappa)\mapsto j'$ maps a CN $\mathsf{c}_j$ to another CN $\mathsf{c}_{j'}$ such that the two CNs are both incident to the same VN $\mathsf{v}_i$ with $i=\phicv(j,\kappa)$.

\begin{figure}[t]
\centering
\begin{tabular}{c}
\begin{tikzpicture}[scale=0.5,every label/.style={font=\footnotesize},]
    \fill[blue!30] (0,0) rectangle (4,4);
    \fill[KITpalegreen!50!green!30] (0,1) rectangle (3,4);
    \draw[step=1, very thick] (0,0) grid (4,4);
    \foreach \row in {0,1,...,3} {
        \foreach \col in {0,1,...,3} {
            \pgfmathtruncatemacro{\idx}{4*(3-\row)+\col+1}
            \node at (\col+0.5,\row+0.5) {\footnotesize $\mathsf{v}_{\idx}$};
        }
    }
    \foreach \i in {5,6,7,8} {
    \node[draw=none] at (\i-4.5,4.5) {$\mathsf{c}_{\i}$};
    }
    \foreach \i in {1,2,3,4} {
    \node[draw=none] at (4.5,4.5-\i) {$\mathsf{c}_{\i}$};
    }

\end{tikzpicture}
\\

\footnotesize (a)  An example PC with $[n=4,k=3] $ component code \\
\\
\begin{tikzpicture}[
    vnode/.style={circle, draw, very thick, minimum size=2mm,inner sep=0pt},
    cnode/.style={rectangle, draw, fill=black, minimum size=2mm},
    every label/.style={font=\footnotesize},
    scale=0.55
]

\foreach \i in {1,2,3,5,6,7,9,10,11} {
    \node[vnode,label=above:{$\mathsf{v}_{\i}$},fill=KITpalegreen!50!green!30] (v\i) at (\i,2) {};
}

\foreach \i in {4,8,12,13,14,15,16} {
    \node[vnode,label=above:{$\mathsf{v}_{\i}$},fill=blue!30] (v\i) at (\i,2) {};
}

\foreach \i in {1,...,4} {
    \node[cnode,label=below:{$\mathsf{c}_{\i}$}] (c\i) at (\i+3,0) {};
}
\foreach \i in {5,...,8} {
    \node[cnode,label=below:{$\mathsf{c}_{\i}$}] (c\i) at (\i+5,0) {};
}

\foreach \k in {1,...,4} {
    \foreach \i in {1,...,4} {
        \pgfmathtruncatemacro{\vidx}{4*(\k-1)+\i}
        \draw (v\vidx) -- (c\k);
    }
}

\foreach \k in {1,...,4} {
    \foreach \i in {1,...,4} {
        \pgfmathtruncatemacro{\vidx}{4*(\i-1)+\k}
        \pgfmathtruncatemacro{\cidx}{4+\k}
        \draw (v\vidx) -- (c\cidx);
    }
}
\end{tikzpicture}\\
\footnotesize (b) Generalized Tanner graph of this PC
\end{tabular}

\caption{Representation of a PC from the data-buffer perspective in (a) and the generalized Tanner graph perspective in (b).}
\label{fig:PCillustration}
\end{figure}
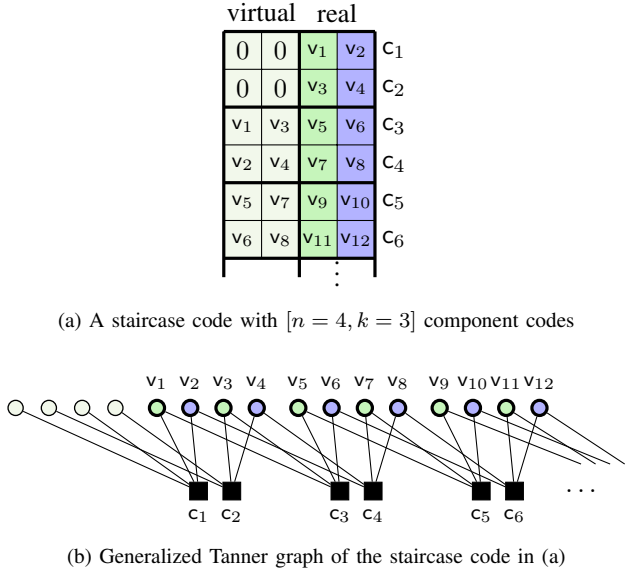
\begin{figure}[t]
\centering
\begin{tabular}{c}
\begin{tikzpicture}[scale=0.5,every label/.style={font=\footnotesize},]
    \foreach \i in {1,2,3,4,5,6} {
    \node[draw=none] at (4.5,6.5-\i) {$\mathsf{c}_{\i}$};
    }

    \node[draw=none] at (1,6.5) {virtual};
    \node[draw=none] at (3,6.5) {real};
    \fill[blue!30] (3,0) rectangle (4,6);
    \fill[KITpalegreen!10] (0,0) rectangle (2,6);
    \fill[KITpalegreen!50!green!30] (2,0) rectangle (3,6);
    
    \draw[very thick] (0,-0.5)--(0,6);
    \draw[very thick] (4,-0.5)--(4,6);
    \draw[very thick] (0,0)--(4,0);
    \draw[very thick] (0,2)--(4,2);
    \draw[very thick] (0,4)--(4,4);
    \draw[very thick] (0,6)--(4,6);
    \draw[very thick] (2,-0.5)--(2,6);

    \draw[step=1] (0,0) grid (4,6);
    \foreach \row in {0,1,...,5} {
        \foreach \col in {0,1} {
            \pgfmathtruncatemacro{\idx}{2*(5-\row)+\col+1}
            \node at (\col+2.5,\row+0.5) {\footnotesize  $\mathsf{v}_{\idx}$};
        }
    }

\foreach \row in {4,5} {
    \foreach \col in {0,1} {
        \node at (\col+0.5,\row+0.5) {$0$};
    }
}

\foreach \x in {0,1} {
        \foreach \y in {0,1} {
            \pgfmathtruncatemacro{\idx}{4*\x+\y+1}
            \node at (0.5,3.5-\y-2*\x) {\footnotesize  $\mathsf{v}_{\idx}$};
            \pgfmathtruncatemacro{\idxx}{4*\x+\y+3}
            \node at (1.5,3.5-\y-2*\x) { \footnotesize  $\mathsf{v}_{\idxx}$};
        }
    }
\node[draw=none] at (3,-0.2) {$\vdots$};

\end{tikzpicture}\\
 \footnotesize (a) A staircase code with $[n=4,k=3]$ component codes \\
\\

\begin{tikzpicture}[
    vnode/.style={circle, draw, very thick, minimum size=2mm, inner sep=0pt},
    cnode/.style={rectangle, draw, fill=black, minimum size=2mm},
    every label/.style={font=\footnotesize},
    scale=0.55
]

\def\xvsep{0.8}   %
\def\xgsep{3.4}    %
\def\yvn{2.0}      %
\def\ycn{0.0}      %
\def\xcshift{1.0} %
\def\xdots{11.0}   %

\foreach \i in {1,...,3} {
    \foreach \j in {1,3} {
        \pgfmathtruncatemacro{\vidx}{4*(\i-1)+\j}
        \pgfmathsetmacro{\x}{(\i-1)*\xgsep + \j*\xvsep}
        \node[vnode,label=above:{$\mathsf{v}_{\vidx}$},fill=KITpalegreen!50!green!30] (v\vidx) at (\x,\yvn) {};
    }
}

\foreach \i in {1,...,3} {
    \foreach \j in {2,4} {
        \pgfmathtruncatemacro{\vidx}{4*(\i-1)+\j}
        \pgfmathsetmacro{\x}{(\i-1)*\xgsep + \j*\xvsep}
        \node[vnode,label=above:{$\mathsf{v}_{\vidx}$},fill=blue!30] (v\vidx) at (\x,\yvn) {};
    }
}

\foreach \j in {1,...,4} {
    \pgfmathtruncatemacro{\vidx}{\j}
    \pgfmathsetmacro{\x}{-\xgsep + \j*\xvsep}
    \node[vnode,fill=KITpalegreen!10,line width = 0.5pt] (vv\vidx) at (\x,\yvn) {};
}

\foreach \i in {1,...,3} {
    \foreach \j in {1,2} {
        \pgfmathtruncatemacro{\cidx}{2*(\i-1)+\j}
        \pgfmathsetmacro{\x}{(\i-1)*\xgsep + \j*\xvsep + \xcshift}
        \node[cnode,label=below:{$\mathsf{c}_{\cidx}$}] (c\cidx) at (\x,\ycn) {};
    }
}

\foreach \k in {1,...,6} {
    \foreach \i in {1,2} {
        \pgfmathtruncatemacro{\vidx}{2*(\k-1)+\i}
        \draw (v\vidx) -- (c\k);
    }
}

\draw (v1) -- (c3);
\draw (v3) -- (c3);
\draw (v2) -- (c4);
\draw (v4) -- (c4);

\draw (v5) -- (c5);
\draw (v7) -- (c5);
\draw (v6) -- (c6);
\draw (v8) -- (c6);

\draw (vv1) -- (c1);
\draw (vv3) -- (c1);
\draw (vv2) -- (c2);
\draw (vv4) -- (c2);

\draw (v9)  -- (\xdots,\ycn+0.65);
\draw (v11) -- (\xdots+0.35,\ycn+0.65);
\draw (v10) -- (\xdots+0.70,\ycn+0.65);
\draw (v12) -- (\xdots+1.05,\ycn+0.65);

\node[draw=none] at (\xdots,\ycn) {$\ldots$};

\end{tikzpicture}\\
\footnotesize  (b) Generalized Tanner graph of the staircase code in (a)
\end{tabular}

\caption{Representation of a staircase code from the data-buffer perspective in (a) and the generalized Tanner graph perspective in (b).}
\label{fig:SCCillustration}
\end{figure}

\subsubsection{Product Codes (PCs)}\label{sec:pc}
A \ac{PC} is a 2D array whose rows and columns are codewords of a component code. A PC with $[n,k,\dmin]$ component code has parameters ${N=n^2}, {K=k^2}, {\dGPC=\dmin^2}$.
Fig.~\ref{fig:PCillustration} shows a PC with ${[n=4,k=3]}$ component codes and its Tanner graph. Each row CN connects to column CNs through $n$ degree-two VNs and vice versa. We index the bit in row $a$ and column $b$ as $\mathsf{v}_{(a-1)n+b}$, the row CNs from $1$ to $n$, and the column CNs from $n+1$ to $2n$. Thus,

\[\phicv(j,\kappa) = \begin{cases}
                         (j-1)n+\kappa & \text{if } j\leq n\\
                         (\kappa-1) n + j-n& \text{if } j> n\\
\end{cases}\]
and
\[\phicc(j,\kappa)= \begin{cases}
    \kappa+n & \text{if } j\leq n\\
    \kappa & \text{if } j> n.\\
\end{cases}\]

\subsubsection{Staircase Codes (SCCs)}\label{sec:scc}
A prominent \ac{SC} GPC is the staircase code~\cite{staircaseCode}, described here in the zipper-code framework~\cite{sukmadji2022zipper}. Fig.~\ref{fig:SCCillustration}\nobreakdash-(a) shows its data-buffer representation for a \([4,3]\) component code.
The buffer consists of a \emph{real buffer}, shown on the right, whose bits are transmitted over the physical channel, and a \emph{virtual buffer}, shown on the left, whose bits are not transmitted but are deterministically and causally mapped from the real buffer. This mapping ensures that each row of the combined buffer forms a valid component codeword. All-zero bits are inserted at the beginning of the virtual buffer to initialize the encoder.
During transmission, new information bits, shown in dark green, are appended to the real buffer and encoded jointly with the mapped virtual bits, shown in light green, to generate the parity bits, shown in blue. This yields a streaming construction for arbitrarily long data sequences. The rate of the staircase code is therefore $R=2\frac{k}{n}-1.$ Fig.~\ref{fig:SCCillustration}\nobreakdash-(b) depicts the corresponding generalized Tanner graph with an infinite and periodic structure.

In a staircase code, blocks of data bits in the real buffer are transposed and copied into the subsequent block of the virtual buffer. To be precise, let $\omega=\frac{n}{2}$ be the width of both the real and virtual buffers. Let $\mathsf{c}_j$, with $j\in \mathbb{N}_{+}$, denote the CN corresponding to the $j$-th row of the data buffer. Then, the VNs in the $j$-th row and $i$-th column of the real buffer are indexed as $\mathsf{v}_{(j-1)\omega+i}$. For the $\kappa$-th incident VN of the $j$-th CN, let $x=\lfloor \frac{j-1}{\omega} \rfloor $ and $y=(j-1)\mod \omega$. We have

\[\phicc(j,\kappa)= \begin{cases}
                        (x-1)\omega+\kappa &\text{if } \kappa \leq \omega\\
    (x+1)\omega + \kappa -\omega & \text{if } \kappa > \omega,
\end{cases}\]
and
 \[\phicv(j,\kappa) = \begin{cases}
     y-\omega + \omega \cdot \phicc(j,\kappa) +1 & \text{if } \kappa \leq \omega\\
     \omega(j-1) + \kappa -\omega & \text{if } \kappa > \omega.
 \end{cases}\]

\section{IBDD of GPCs}\label{sec:iBDD}

\subsection{Algorithm Description}
For a block \ac{GPC}, one \ac{iBDD} iteration consists of decoding each \ac{CN} once using \ac{BDD}. The \ac{iBDD} algorithm stops when either all syndromes are zero, indicating that no correctable errors remain, or when a predefined maximum number of iterations \(L\) is reached.

To simplify the implementation, the syndrome of each \ac{CN} word is typically stored, since \ac{BDD} requires only the syndrome to estimate the error locations. At each \ac{CN} decoding step, up to \(t\) data bits are flipped, and the affected syndromes are updated accordingly.

For \ac{SC} GPCs, a windowed decoding schedule is employed to accommodate their streaming nature. Specifically, \(W\) \acp{CN} are grouped to form a decoding window. Similar to the block GPC case, the \acp{CN} within this window are decoded over \(L\) iterations. Then, the window is shifted forward by an integer multiple of \(\omega\). For staircase codes, we consider a shift of \(\omega\) CNs (one block). The data bits associated with the first \(\omega\) \acp{CN} in the window are output and the iterative decoding process continues.

The \ac{iBDD} algorithms for block and \ac{SC} GPCs are described in Algorithm~\ref{alg:iBDD} and Algorithm~\ref{alg:iBDDwindow}, respectively.

\begin{algorithm}[t]
    \footnotesize
    \DontPrintSemicolon
    \caption{IBDD of GPCs}\label{alg:iBDD}
    \textbf{Input}: $\bm{Y}\in\{0,1\}^{N}$\;

    \tcp{Initialize syndrome buffer}
    \For{$j=1,2,\ldots,M$}{
        $\bm{S}_j \gets \textsc{ComputeSyndrome}(\mathsf{c}_j)$ \;
    }

    \tcp{Iterative decoding}
    \For(\tcp*[f]{$L$: number of iterations}){$\ell=1,2,\ldots,L$}{
        \For(\tcp*[f]{$j$: index of the CN}){$j=1,2,\ldots,M$}{
            [$\mathcal{L}_e, n_e$] $\gets \text{BDD}(\bm{S}_j)$\;
            \If(\tcp*[f]{$n_e$: number of errors}){$n_e > 0$}{
                \For(\tcp*[f]{$\mathcal{L}_e$: error locations}){$e \in \mathcal{L}_e$}{
                \tcp{flip data bit}
                    $Y_{\phicv(j,e)} \gets Y_{\phicv(j,e)}  \oplus 1$\; 
                     \tcp{update syndrome register}
                    $\bm{S}_{\phicc(j,e)}\gets \textsc{ComputeSyndrome}(\mathsf{c}_{\phicc(j,e)})$\;
                }
            }

        }
        \If{\(\bm{S}_{j}=\bm{0}\) for all \(j\)}{\Return $\bm{Y}$}
    }
    \textbf{Output}: $\bm{Y}\in\{0,1\}^{N}$\;
\end{algorithm}

\begin{algorithm}[t]
    \footnotesize
    \DontPrintSemicolon
    \caption{Windowed iBDD of SC GPCs}\label{alg:iBDDwindow}
    $i \gets 1$\tcp*[f]{$i$: starting CN index of the window} \;
    \While{True}{
        \For(\tcp*[f]{$L$: number of iterations}){$\ell=1,2,\ldots,L$}{
            Block iBDD for $\mathsf{c}_j, j \in \{i,i+1,\ldots, i + W-1\}$\;
        }
        $i \gets i +\omega$\;
    }

\end{algorithm}

\subsection{Complexity Analysis}
Based on Algorithm~\ref{alg:iBDD}, we discuss the low complexity of \ac{iBDD} from three perspectives.

\subsubsection{Computational Complexity}
Upon initialization, the syndromes of all \acp{CN} are computed with complexity \(\mathcal{O}(N)\). Each subsequent iteration has complexity \(\mathcal{O}(M)\), since each \ac{CN} performs only a small, constant number of operations to decode its syndrome, as described in Sec.~\ref{sec:BDDforBCH}. Note that the \ac{BDD} step for \acp{CN} with a zero syndrome is skipped, and the actual number of \ac{BDD} executions should be evaluated by simulation.

\subsubsection{Memory}
A data buffer (\(N\) bits) and, typically, a syndrome buffer (\(btM\) bits) are required, e.g., as in~\cite{FougstedtVLSI}. In total, the required memory is \(M_{\text{iBDD}}= N + btM\) bits. Since \ac{iBDD} is an intrinsic message-passing decoder, no additional memory is needed to store the channel output for later use.

\subsubsection{Internal Decoder Data Flow}
Upon initialization, all bits are read once to compute the syndromes. During iterative decoding, most of the data buffer remains unchanged and is updated only when a \ac{BDD} correction occurs, which affects at most \(t\) bits per component codeword in each iteration. In contrast, the syndrome buffer is accessed frequently for both read and write operations. However, the syndrome buffer is significantly smaller than the data buffer.

\section{Error-and-Erasure Decoding (EaED)}\label{sec:EaED}

This section describes the \ac{EaED} algorithm with miscorrection detection, which is a key component of the \acs{RDRSD}. Since \ac{EaED} serves as a component-code decoder, we consider only a component (e)BCH code when referring to (code)words in this section.

\subsection{EaED with Miscorrection Detection}
This algorithm improves upon the versions in~\cite{rapp2021error,miao2022JLT} by enabling a list-based decoding option, which allows a larger number of errors and erasures to be corrected.

Let $\bm{x}\in\CW$ be the transmitted codeword and $\bm{y}\in\ZQO^n$ its noisy received version.
The number of erasures in $\bm{y}$ is denoted by $E$ and
the decoder output is denoted by ${\bm{w}\coloneqq\text{EaED}(\bm{y})}$.
We further introduce a design parameter ${\mathcal{J}\in \mathbb{N}_{+}}$ that controls the performance--complexity trade-off. %

First, initialize an empty candidate list \(\mathcal{W}=\emptyset\). If \(E=0\), set \(J=1\). If \(E>0\), set
\[
J\coloneqq \min\{2^{E-1},\mathcal{J}\}.
\]
Then, $J$ independent pairs of complementary random vectors $\pone,\ptwo\in\ZO^{E}$ satisfying $\pone+\ptwo=(1,1,\ldots,1)$ are generated. Each pair $(\pone,\ptwo)$ defines a \emph{filling pattern} for the erasures.

We then apply the following procedure to each test pattern.
By replacing the erasures in $\bm{y}$ using $\pone$ and $\ptwo$, we obtain two test patterns $\yone,\ytwo\in\ZO^n$ for each pair. Note that $\bm{y} = \yone = \ytwo$ for $E=0$.
Next, we perform a \ac{BDD} step for each test pattern. Specifically, for each generated test pattern indexed by \(i\), we compute $\wi\coloneqq\BDD(\bm{y}^{(i)})$. If $\wi\in\CW$, we perform miscorrection detection using \eqref{eq:miscorrectionCheck}, as described later in Sec.~\ref{subsec:miscorrectionDetection}. If $\wi$ is classified as a miscorrection, it is discarded. Otherwise, we compute its distance to $\bm{y}$ over the non-erased positions of $\bm{y}$, i.e.,
\(
    \dH=\dnE{\bm{y}} \left(\bm{y}, \wi\right),
\)
and add the tuple $(\wi,\dH)$ to the candidate list $\mathcal{W}$. This procedure is repeated for all \(J\) filling-pattern pairs.

Finally, if $\mathcal{W}=\emptyset$, the decoder declares a failure and outputs $\bm{w}=\bm{y}$. Otherwise, the decoder outputs the decision codeword
\({\text{argmin}}_{\bm{w}\in \mathcal{W}} \, \dnE{\bm{y}} \left(\bm{y}, \bm{w}\right)\). Ties are broken at random.

The complete \ac{EaED} procedure is summarized in Algorithm~\ref{alg:eaed}. Note that when $E=0$ and no miscorrection detection is performed, the \ac{EaED} reduces to conventional \ac{BDD}.

\begin{algorithm}[t]
\footnotesize
\DontPrintSemicolon
\caption{EaE decoder (EaED)}
\label{alg:eaed}

\textbf{Input}: $\bm{y}\in\ZQO^n$\;

$E\gets$ $\E(\bm{y})$\tcp*[r]{Number of erasures in $\bm{y}$}
$\mathcal{W} \gets \emptyset$\tcp*[r]{Set of candidate codewords}
$\mathcal{F}\gets \emptyset$\tcp*[r]{Set of conflicting bits}
\If{$E>\dmin$}{
    return $\{\bm{y}\}$, $\mathcal{F}$
}

\For{$J=1,2,\ldots,\max\{1,\min\{2^{E-1},\mathcal{J}\}\}$}{
    $\pone,\ptwo \gets$ two random complementary vectors in $\ZO^{E}$\;
    
    $\bm{y}^{(i)} \gets \bm{y}$ with erasures replaced by $\bm{p}^{(i)}$\;
    
    \For{$i=\{1,\ldots, 1+\mathbbm{1}_{E>0}\}$}{
        $\bm{w}^{(i)} \gets \BDD(\bm{y}^{(i)})$\;
        
        \If{$\bm{w}^{(i)} \in \CW$}{
            \If{miscorrection detection}{
            Check $\bm{w}^{(i)}$ with \eqref{eq:miscorrectionCheck}\;
                $\mathcal{F}\gets\{j:y_j\neq \que, w^{(i)}_{j}\neq y_j\}\cup \mathcal{F}$ \;
                \If{$\bm{w}^{(i)}$ is detected as a miscorrection}
                {continue\;}
            }
            $\dH_i\gets \dnE{\bm{y}} (\bm{y}, \wi)$\;
            add $(\bm{w}^{(i)},\dH_i)$ to $\mathcal{W}$\;
        }
    }
}

\eIf{$\mathcal{W}=\emptyset$}{
    $\bm{w} \gets \bm{y}$\tcp*[r]{Decoding failure}
}{
    select $\bm{w}={\text{argmin}}_{\bm{c}\in \mathcal{W}} \, \dnE{\bm{y}} \left(\bm{y}, \bm{c}\right)$;
}

\textbf{Output}: $\bm{w}\in \CW \cup \{\bm{y}\}$, $\mathcal{F}$\;
\end{algorithm}

Another commonly used \ac{EaE} decoder is the one-step algebraic \ac{EaE} decoding algorithm proposed by Forney~\cite{forney1965decoding}. It extends the Gorenstein--Zierler algorithm~\cite[Ch.~6]{roth_introduction_2006} and corrects errors and erasures jointly by solving the key equation. Forney's \ac{EaE} decoder follows the decoding rule
\begin{equation*}
    \bm{w}=
    \begin{cases}
    \bm{c}, & \exists\, \bm{c}\in \CW \text{ such that } \E(\bm{y}) + 2\dnE{\bm{y}}(\bm{y},\bm{c}) < \dmin,\\
    \bm{y}, & \text{otherwise}.
    \end{cases}
\end{equation*}
We do not use Forney's algebraic \ac{EaE} decoder as the component-code decoder, since it is limited by the minimum-distance decoding radius and, based on our observations, performs worse than the proposed \ac{EaED} in iterative decoding when soft information is available~\cite{rapp2021error}.

For reference, we also define an ideal EaED decoder that achieves miscorrection-free decoding by appending genie-aided miscorrection detection to EaED, i.e.,
\[
\text{idealEaED}(\bm{y}) =
\begin{cases}
    \text{EaED}(\bm{y}) & \text{if } \text{EaED}(\bm{y}) = \bm{x},\\
    \mathsf{fail} & \text{if } \text{EaED}(\bm{y}) \neq \bm{x},
\end{cases}
\]
where $\bm{x}$ denotes the transmitted component-code codeword.

For clarity, we define several iterative decoding algorithms for GPCs. A decoder that uses an \ac{EaED} without miscorrection detection as its component decoder is referred to as \emph{\ac{iEaED}}. Replacing this component decoder with an ideal \ac{EaED} yields the \emph{ideal~\ac{iEaED}}. When $\te=0$, the \emph{ideal \ac{iEaED}} reduces to the \emph{ideal \ac{iBDD}}.

\subsection{Miscorrection Analysis}\label{sec:miscorrection}
A miscorrection occurs when the \ac{BDD} outputs an incorrect codeword that is closer to the received word than to the transmitted codeword.
This is more likely for component codes with small minimum distance, where codewords are relatively close to each other.
Miscorrections are therefore particularly detrimental to iterative hard-decision decoding of high-rate \acp{GPC}.
Typically, when \ac{BDD} is applied to a word with more than $t$ errors, a miscorrection increases the number of erroneous bits by $t$.
Moreover, the miscorrected word is declared as a valid codeword with zero syndrome and can no longer be detected by the associated CN.
This may eventually lead to a stall pattern and cause an iterative decoding failure.

The miscorrection probability for \ac{BDD} was derived in~\cite{mceliece1986decoder} and extended to EaED in~\cite{Miao26IZS}.
We give an overview of the results here. We use $U$ and $E$ to denote the number of errors and erasures in a word, respectively.
\setlength{\tabcolsep}{2pt}
\begin{table}[t]
    \centering
    \caption{Decoding success probability for different numbers of errors~\(U\) and erasures~\(E\) for the \([256,239,6]\) eBCH code.}
    \label{tab:decoding_success}
    \begin{tabular}{c|ccccccccccc}
        \hline
        \diagbox[width=3.5em,height=2em]{\footnotesize $U$}{\footnotesize $E$}
        & 0 & 1 & 2 & 3 & 4 & 5 & 6 & 7 & 8 & 9 & 10 \\
        \hline
        0
        & 1 & 1 & 1 & 1 & 1 & 1 & 0.69 & 0.45 & 0.29 & 0.18 & 0.11 \\
        1
        & 1 & 1 & 1 & 1 & 0.62 & 0.37 & 0.22 & 0.12 & 0.069 & 0.039 & 0.021 \\
        2
        & 1 & 1 & 0.38 & 0.25 & 0.093 & 0.062 & 0.023 & 0.016 & 0.0057 & 0.0039 & 0.0014 \\
        3 & 0 & 0&0& 0 & 0&0& 0 & 0&0& 0 & 0\\
        \hline
    \end{tabular}
\end{table}

\begin{table}[t]
    \centering
    \caption{Miscorrection probability for different numbers of errors~\(U\) and erasures~\(E\) for the \([256,239,6]\) eBCH code.}
    \label{tab:miscorrection}
    \begin{tabular}{c|ccccccccccc}
        \hline
        \diagbox[width=3.5em,height=2em]{\footnotesize $U$}{\footnotesize $E$}
        & 0 & 1 & 2 & 3 & 4 & 5 & 6 & 7 & 8 & 9 & 10 \\
        \hline
        0
        & 0 & 0 & 0 & 0 & 0 & 0 & 0 & 0.27 & 0.21 & 0.41 & 0.31 \\
        1
        & 0 & 0 & 0 & 0 & 0.0012 & 0.31 & 0.24 & 0.44 & 0.33 & 0.48 & 0.36 \\
        2
        & 0 & 0 & 0.12 & 0.37 & 0.31 & 0.47 & 0.36 & 0.49 & 0.37 & 0.50 & 0.38 \\
        3
        & 0 & 0.49 & 0.37 & 0.50 & 0.38 & 0.50 & 0.38 & 0.50 & 0.38 & 0.50 & 0.38 \\
        4
        & 0.49 & 0.50 & 0.38 & 0.50 & 0.38 & 0.50 & 0.38 & 0.50 & 0.38 & 0.50 & 0.38 \\
        5
        & 0.0039 & 0.50 & 0.38 & 0.50 & 0.38 & 0.50 & 0.38 & 0.50 & 0.38 & 0.50 & 0.38 \\
        \hline
    \end{tabular}
\end{table}

\subsubsection{BDD}
For primitive \ac{BCH} codes, the miscorrection probability is zero for \(U \leq t\) and is approximately \(\frac{1}{t!}\) for \(U>t\).

For e\ac{BCH} codes, odd-weight codewords of weight \(w\) in the primitive code are mapped to codewords of weight \(w+1\) in the extended code.
Using this property, and by slightly extending the results of~\cite{mceliece1986decoder}, it was shown in~\cite{Miao26IZS} that the miscorrection probability is zero for \(U \leq t\).
For \(U>t\), the value is nearly zero if \(U+t\) is odd. If \(U+t\) is even, it is approximately \(\frac{1}{t!}\).

\subsubsection{EaED}

We consider the commonly used case with ${\mathcal{J}=1}$.

For primitive BCH codes, the miscorrection probability is zero when \( 2U + E < \dmin \). When \( 2U + E \gg \dmin \), the miscorrection probability is close to \( {1 - (1-\left(\frac{1}{t!}\right))^2 }\). For values of \( 2U + E \) that are slightly above \( \dmin \), the miscorrection probability increases from zero to \( {1 - (1-\left(\frac{1}{t!}\right))^2 }\).

For eBCH codes, the results do not admit a simple closed-form approximation. In Tab.~\ref{tab:decoding_success} and Tab.~\ref{tab:miscorrection}, we list the decoding success and miscorrection probabilities when decoding the \([256,239,6]\) eBCH code with \( U \) errors and \( E \) erasures. As observed from the numerical results, the decoding success probability does not drop abruptly to zero once \(2U+E\) exceeds \(\dmin-1\); instead, it decreases from $1$ to $0$. However, for large values of \(2U+E\), the miscorrection probability increases compared to that of conventional \ac{BDD}, which is listed in the first column of Tab.~\ref{tab:miscorrection}.

The decoding success-rate analysis motivates lines~5--6 of Algorithm~\ref{alg:eaed}, which allow decoding when \(E=\dmin\), a case not permitted in conventional \ac{EaE} decoding~\cite{MoonBook}. The success and miscorrection rates show that \ac{EaED} provides a low-complexity method for decoding beyond the minimum distance using \ac{BDD}, albeit at the cost of an increased miscorrection probability. This highlights the need for effective miscorrection detection by exploiting soft information.

\section{\acs{RDRSD} Algorithm}\label{sec:DRSD}
In this section, we present the \acs{RDRSD}, which improves upon the conventional \ac{DRSD}~\cite{miao2022JLT}. The block and windowed versions are summarized in Algorithm~\ref{alg:DRSD} and Algorithm~\ref{alg:DRSDwindow}, respectively, and described in detail below.

\subsection{Refined Dynamic Reliability Scores (DRSs)}

The \ac{DRS} serves a role similar to soft information, such as \acp{LLR}, in soft-decision decoding algorithms. Specifically, it assigns a reliability score to each bit in a GPC, with higher scores indicating higher reliability. We denote the vector of bit \ac{DRS} values by \(\bm{D}=(D_1,D_2,\ldots,D_N)\).

The entries of \(\bm{D}\) are quantized using \(q\) bits, giving the score range \(\{0,1,\ldots,2^q-1\}\). They are updated through hard messages, i.e., one-bit signals that indicate, for example, whether a score should be increased by one or left unchanged. Similarly to \ac{iBDD}, the \ac{DRS} uses an intrinsic message-passing mechanism: after initialization, the channel output is not stored separately. Once the channel observations have been converted into \(\bm{D}\), decoding proceeds solely based on the \ac{DRS} values, the ternary data buffer, and the syndrome buffer, which are updated directly throughout the iterations.

\subsubsection{Initialization} \label{sec:DRSini}
The bit \ac{DRS} \(D_i\) is initialized using precomputed thresholds. If \(|\tilde{y}_i|\in[T_{d},T_{d+1})\), where \(T_{d}\) and \(T_{d+1}\) are the thresholds associated with \(d\) and \(d+1\), respectively, then \(D_i=d\).
Following~\cite{miao2022JLT,Rapp24ECOC}, these thresholds are chosen such that the initial values of \(D_i\) are assigned to asymptotically equal-sized reliability classes. Let $\lambda = D_{\max} - D_{\min}+1$ denote the number of initial reliability classes.
With channel symmetry and equiprobable inputs, it is sufficient to compute the threshold probabilities from the distribution of \(|\tilde{Y}|\).
The thresholds are chosen to satisfy
\(
    \Pr\!\big(T_{d}\leq |\tilde{Y}| < T_{d+1}\big)
    = \frac{1}{\lambda}.
\)
For the BI-AWGN channel, this condition becomes
\[
\begin{aligned}
&Q\!\big(\tfrac{T_d-1}{\sigma}\big)
-
Q\!\big(\tfrac{T_{d+1}-1}{\sigma}\big)
+
Q\!\big(\tfrac{T_d+1}{\sigma}\big)
-
Q\!\big(\tfrac{T_{d+1}+1}{\sigma}\big)
= \tfrac{1}{\lambda},
\end{aligned}
\]
with $T_{D_{\min}}=0$ and $T_{D_{\max}+1}=\infty.$
In addition, for the \acs{RDRSD}, we propose to use a \emph{CN reliability score}, or \emph{CN DRS}, which quantifies the reliability of a CN after decoding. We denote the vector of CN DRS values by \(\bm{D}_{\text{CN}}=(D_{\text{CN},1},D_{\text{CN},2},\ldots,D_{\text{CN},M})\).
In \acp{GPC}, a decoded codeword that passes the miscorrection check is considered highly reliable~\cite{hager2018approaching}.
We therefore assign a relatively high initial value \(d_0\) to the corresponding entry of \(\bm{D}_{\text{CN}}\) and increase it further for each iteration in which the decoded codeword remains unchanged. Here, \(d_0\) is code-dependent.

During decoding, the reliability value used for the \(\kappa\)-th bit of the currently decoded CN \(\mathsf{c}_j\) is \(D_{\text{CN},\phicc(j,\kappa)}\) if the CN \(\mathsf{c}_{\phicc(j,\kappa)}\) is already a valid codeword. Otherwise, it is \(D_{\phicv(j,\kappa)}\).

\subsubsection{Miscorrection Detection with DRS}\label{subsec:miscorrectionDetection}
The key idea inherited from \ac{DRSD} is to exploit the soft information provided by the \ac{DRS} to assist \ac{EaED}, so that its behavior approaches that of an ideal \ac{EaED} in which all miscorrections are discarded.
As shown in Algorithm~\ref{alg:eaed}, whenever \ac{BDD} outputs a valid codeword, a \ac{DRS}-based miscorrection detection step is performed.
This step is described in detail below.

The first approach for miscorrection detection is based on anchor bits, i.e., highly reliable bits whose \(D_i\) values exceed a threshold \(\ta\), where \(\ta\) is an optimizable parameter.
Any \ac{BDD} output that conflicts with an anchor bit, i.e., attempts to flip it, is declared a miscorrection and discarded.
As shown in~\cite{Miao26IZS}, this simple approach can keep the miscorrection probability below \(0.01\) under moderate channel conditions.

An alternative approach can be used in addition to the anchor bits. For a \ac{BDD} output $\bm{w}$, we compute the sum of the \(D_i\) values over the flipped positions, i.e.,
    $\sum_{i:w_i\neq y_i,\, y_i\neq \que} D_i$.
If this sum exceeds a threshold $\tee$, the codeword is also declared a miscorrection and discarded.
\begin{algorithm}[t]
    \footnotesize
    \DontPrintSemicolon
    \caption{\acs{RDRSD} of block GPCs}\label{alg:DRSD}
    \textbf{Input}: $\tilde{\bm{Y}}\in\mathbb{R}^{N}$\;

$\bm{Y}\gets \bm{0}^{N}$, $\bm{D}\gets \bm{0}^{N}$, $\bm{D}_{\text{CN}}\gets d_0\cdot\bm{1}^{M}$\;

    \For{$i=1,2,\ldots,N$}{
    \tcp{Initialize data buffer}
    \lIf{$|\tilde{y}_{i}|\leq \te$}{$y_{i}=\que$}
    \lElseIf{$\tilde{y}_{i}>\te$}{$y_{i}=0$}
    \lElse{$y_{i}=1$}
    \tcp{Initialize bit DRSs}
    \lIf{$T_{d}\leq|\tilde{y}_{i}|<T_{d+1}$}{$D_{i}=d$}
    }

        \tcp{Initialize syndrome buffer}
    \For{$j=1,2,\ldots,M$}{
        $\bm{S}_j \gets \textsc{ComputeSyndrome}(\mathsf{c}_j)$\;
    }

$\bm{D}_{\text{CN}}\gets$ initial CN DRSs $d_0$\;
\tcp{Iterative decoding}
\For(\tcp*[f]{$L$: number of iterations}){$\ell=1,2,\ldots,L$}{
  \For{$j=1,2,\ldots,M$}{
    \uIf{$\mathsf{c}_j$ is already a codeword}{
      $D_{\text{CN},j}\gets D_{\text{CN},j}+1$\;
    }
    \Else{
	      $[\bm{w},\mathcal{F}]\gets $decode $\mathsf{c}_j$ with Algorithm~\ref{alg:eaed}\;
	      \If{Decoding success}{
	        \For{$e\in \{e:w_e\neq y_e\}$}{
	          $Y_{\phicv(j,e)} \gets w_{e}$\;
	          $\bm{S}_{\phicc(j,e)}\gets \textsc{ComputeSyndrome}(\mathsf{c}_{\phicc(j,e)})$\;
	        }
	    
	      }
      \For{$\kappa \in \mathcal{F}$}{
        reduce DRS of $\kappa$-th bit of CN\;
      }

    }

  }
  \If{all syndrome $\bm{S}_j=\bm{0}$ and no erasures are left}{
      \Return $\bm{Y}$\;
  }
  update $\ta$ and $\tee$ based on $\ell$\;
}
    Fill erasures in \(\bm{Y}\) with i.i.d. Bernoulli-\(\frac{1}{2}\) bits\;
    \textbf{Output}: $\bm{Y}\in\{0,1\}^{N}$\;
\end{algorithm}

\begin{algorithm}[t]
    \footnotesize
    \DontPrintSemicolon
    \caption{Windowed \acs{RDRSD} of \ac{SC} GPCs}\label{alg:DRSDwindow}
    $i \gets 1$\tcp*[f]{$i$: starting CN index of the window} \;
    \While{True}{
        \For(\tcp*[f]{$L$: number of iterations}){$\ell=1,2,\ldots,L$}{
            Block \acs{RDRSD} for $\mathsf{c}_j, j \in \{i,i+1,\ldots, i + W-1\}$\;
        }
        $i \gets i+\omega$\;
    }

\end{algorithm}

In practice, both miscorrection detection methods can be enabled simultaneously, particularly when $\mathcal{J}>1$. Overall, a \ac{BDD} decision is declared a miscorrection if
\begin{align}\label{eq:miscorrectionCheck}
\{i : w_i \neq y_i,\, y_i \neq \que, D_i > \ta^{(\ell)}\}\neq \emptyset \\ \nonumber
\text{ or } \sum_{i:w_i\neq y_i,\, y_i\neq \que} D_i > \tee^{(\ell)}.
\end{align}
The thresholds $\ta^{(\ell)}$ and $\tee^{(\ell)}$ are updated according to the current iteration number $\ell$. A simple update strategy is to increase both thresholds every five iterations, as in~\cite{miao2022JLT}.

\subsubsection{Updating DRSs}
At the beginning of each CN decoding step for \(\mathsf{c}_j\), if the corresponding CN word is already a valid codeword, \(D_{\text{CN},j}\) is increased by \(1\).
Otherwise, the CN word is decoded using Algorithm~\ref{alg:eaed}.
After the \ac{EaED} decision is made, the \(D_i\) values of the bits that conflict with any \ac{BDD} decision, i.e., the bits in \(\mathcal{F}\), are decreased by \(1\).
All updated entries of \(\bm{D}\) and \(\bm{D}_{\text{CN}}\) are clipped to the range \([0,2^q-1]\).

\subsection{Decoding Algorithm}
We first describe the \acs{RDRSD} algorithm for block \acp{GPC}, which is summarized in Algorithm~\ref{alg:DRSD}.

At the beginning of decoding, the data buffer is initialized with ternary decisions obtained from the channel output.
If \(\tilde{y}\in[-\te,\te]\), an erasure is declared, i.e., \(y=\que\).
Otherwise, the hard decision is given by \(y=\psi(\tilde{y})\in\{0,1\}\).
The syndrome of each CN is then computed by treating the erased bits as zero, i.e., erased bits do not contribute to the syndrome.

In addition to initializing the data and syndrome buffers, \(\bm{D}\) is initialized as described in Sec.~\ref{sec:DRSini}, whereas all entries of \(\bm{D}_{\text{CN}}\) are set to a code-dependent initial value \(d_{0}\).
Decoding then proceeds iteratively.
In each iteration, all CNs decode their associated bits using \ac{EaED} with miscorrection detection.
After each CN decision, the bits in the identified error-location set are flipped, the corresponding syndrome buffers of the two involved component codes are updated, and the affected entries of \(\bm{D}\) and \(\bm{D}_{\text{CN}}\) are updated.
This process is repeated until all syndromes are zero and no erasures remain, or until the maximum number of iterations \(L\) is reached.

For \ac{SC} \acp{GPC}, we employ windowed \acs{RDRSD}, whose routine is similar to that for block \acp{GPC} within a decoding window. The only difference is that the thresholds \(\ta\) and \(\tee\) should be configured based on the effective CN decoding count, determined by the iteration index and the CN position within the window. Algorithm~\ref{alg:DRSDwindow} summarizes the windowed \acs{RDRSD} algorithm.

\subsection{Complexity Analysis}
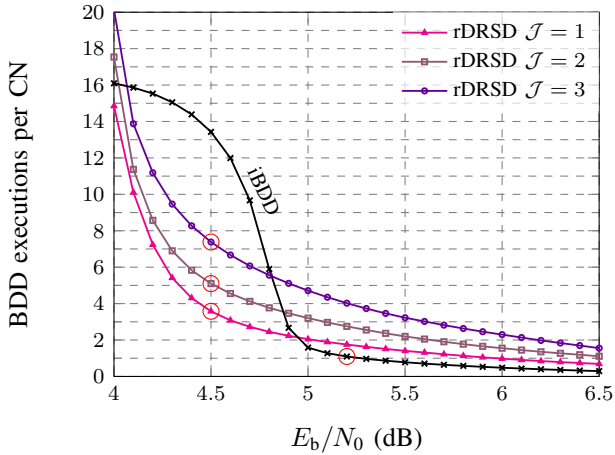
\begin{figure}
    \centering
    \begin{tikzpicture}
\pgfplotsset{grid style={dashed, gray}}
\pgfplotsset{every tick label/.append style={font=\footnotesize}}

\begin{axis}[
    anchor=north west,
    yshift=-0.35cm,
    width=8cm,
    height=6.4cm,
    xmin=4,
    xmax=6.5,
    ymin=0,
    ymax=20,
    xtick={3.5,4,4.5,5,5.5,6,6.5,7},
    ytick={0,2,...,200},
    xlabel={$\Eb/\No$ (dB)},
    ylabel={BDD executions per CN},
    axis background/.style={fill=white},
    xmajorgrids,
    ymajorgrids,
    yminorgrids,
    minor y tick num=1,
    legend style={
    at={((1,1)},
    anchor=north east,
    draw=none,
    fill opacity=0.7,
    text opacity=1,
    legend columns=1,
    row sep=0pt,
    font=\footnotesize
}
]

    \addlegendentry{rDRSD $\mathcal{J}=1$}
    \addlegendentry{rDRSD $\mathcal{J}=2$}
    \addlegendentry{rDRSD $\mathcal{J}=3$}

\addplot [color=magenta, line width=0.7pt, mark=triangle, mark options={fill=white, mark size=1pt}]table[x=EbNo_dB,y=normalizedSumCalls, col sep=semicolon,row sep=crcr] {
EbNo_dB;simulatedFrames;FER;BER;normalizedSumCalls\\ 
3.4;12800;1;0.0309752;11.9414;\\ %
3.5;12800;1;0.0291807;13.2331;\\ %
3.6;12800;1;0.0269423;14.5479;\\ %
3.7;12800;1;0.0240721;15.8276;\\ %
3.8;12800;1;0.0199355;16.9003;\\ %
3.9;12800;0.988125;0.0133617;17.2419;\\ %
4;12800;0.627266;0.00401237;14.8492;\\ %
4.1;12800;0.0455469;0.000130583;10.1084;\\ %
4.2;12800;7.8125e-05;1.88351e-07;7.23315;\\ %
4.3;12800;0;0;5.41505;\\ %
4.4;12800;0;0;4.30662;\\ %
4.5;12800;0;0;3.57783;\\ %
4.6;12800;0;0;3.08136;\\ %
4.7;12800;0;0;2.72544;\\ %
4.8;12800;0;0;2.4507;\\ %
4.9;12800;0;0;2.23329;\\ %
5;12800;0;0;2.04906;\\ %
5.1;12800;0;0;1.89008;\\ %
5.2;12800;0;0;1.75214;\\ %
5.3;12800;0;0;1.62663;\\ %
5.4;12800;0;0;1.51085;\\ %
5.5;12800;0;0;1.40487;\\ %
5.6;12800;0;0;1.31234;\\ %
5.7;12800;0;0;1.2192;\\ %
5.8;12800;0;0;1.13165;\\ %
5.9;12800;0;0;1.04933;\\ %
6;12800;0;0;0.979027;\\ %
6.1;12800;0;0;0.908851;\\ %
6.2;12800;0;0;0.843847;\\ %
6.3;12800;0;0;0.78903;\\ %
6.4;12800;0;0;0.738983;\\ %
6.5;12800;0;0;0.688727;\\ %
6.6;12800;0;0;0.645135;\\ %
6.7;12800;0;0;0.605307;\\ %
6.8;12800;0;0;0.566669;\\ %
6.9;12800;0;0;0.530174;\\ %
7;12800;0;0;0.495808;\\ %
 }; %

    \addplot [color=magenta!50!black, line width=0.7pt, mark=square, mark options={fill=white, mark size=1pt}]table[x=EbNo_dB,y=normalizedSumCalls, col sep=semicolon,row sep=crcr] {
        EbNo_dB;simulatedFrames;FER;BER;normalizedSumCalls\\
        3.4;19200;1;0.0340455;23.5009;\\ %
        3.5;19200;1;0.0320035;26.01;\\ %
        3.6;19200;1;0.02894;28.3457;\\ %
        3.7;19200;1;0.0246101;30.3445;\\ %
        3.8;19200;0.997344;0.0180635;31.2483;\\ %
        3.9;19200;0.80724;0.00776974;27.1351;\\ %
        4;19200;0.170573;0.000782628;17.5382;\\ %
        4.1;19200;0.00260417;6.53982e-06;11.3701;\\ %
        4.2;19200;0;0;8.57133;\\ %
        4.3;19200;0;0;6.89769;\\ %
        4.4;19200;0;0;5.83232;\\ %
        4.5;19200;0;0;5.09662;\\ %
        4.6;19200;0;0;4.55314;\\ %
        4.7;19200;0;0;4.12108;\\ %
        4.8;19200;0;0;3.76764;\\ %
        4.9;19200;0;0;3.46856;\\ %
        5;19200;0;0;3.20157;\\ %
        5.1;19200;0;0;2.96574;\\ %
        5.2;19200;0;0;2.74866;\\ %
        5.3;19200;0;0;2.55086;\\ %
        5.4;19200;0;0;2.36723;\\ %
        5.5;19200;0;0;2.19722;\\ %
        5.6;19200;0;0;2.05071;\\ %
        5.7;19200;0;0;1.90497;\\ %
        5.8;19200;0;0;1.7732;\\ %
        5.9;19200;0;0;1.65348;\\ %
        6;19200;0;0;1.55064;\\ %
        6.1;19200;0;0;1.44895;\\ %
        6.2;19200;0;0;1.35391;\\ %
        6.3;19200;0;0;1.27103;\\ %
        6.4;19200;0;0;1.19201;\\ %
        6.5;19200;0;0;1.10957;\\ %
        6.6;19200;0;0;1.03568;\\ %
        6.7;19200;0;0;0.965262;\\ %
        6.8;19200;0;0;0.895089;\\ %
        6.9;19200;0;0;0.829231;\\ %
        7;19200;0;0;0.765121;\\ %
    };%

    \addplot [color=blue!60!red, line width=0.7pt, mark=o, mark options={fill=white, mark size=1pt}]table[x=EbNo_dB,y=normalizedSumCalls, col sep=semicolon,row sep=crcr] {
        EbNo_dB;simulatedFrames;FER;BER;normalizedSumCalls\\
        3.4;19200;1;0.0357657;44.1581;\\ %
        3.5;19200;1;0.0323668;48.2623;\\ %
        3.6;19200;1;0.0280597;51.5954;\\ %
        3.7;19200;1;0.0227119;52.2372;\\ %
        3.8;19200;0.995521;0.015747;47.5465;\\ %
        3.9;19200;0.769427;0.00640134;34.2624;\\ %
        4;19200;0.142083;0.000618138;20.1567;\\ %
        4.1;19200;0.00208333;5.20905e-06;13.8733;\\ %
        4.2;19200;0;0;11.1851;\\ %
        4.3;19200;5.20833e-05;7.15256e-09;9.46414;\\ %
        4.4;19200;0;0;8.27209;\\ %
        4.5;19200;0;0;7.3808;\\ %
        4.6;19200;0;0;6.6685;\\ %
        4.7;19200;0;0;6.0686;\\ %
        4.8;19200;0;0;5.558;\\ %
        4.9;19200;0;0;5.11063;\\ %
        5;19200;0;0;4.71236;\\ %
        5.1;19200;0;0;4.35084;\\ %
        5.2;19200;0;0;4.0228;\\ %
        5.3;19200;0;0;3.72707;\\ %
        5.4;19200;0;0;3.4623;\\ %
        5.5;19200;0;0;3.22419;\\ %
        5.6;19200;0;0;3.02124;\\ %
        5.7;19200;0;0;2.81942;\\ %
        5.8;19200;0;0;2.6331;\\ %
        5.9;19200;0;0;2.45645;\\ %
        6;19200;0;0;2.2987;\\ %
        6.1;19200;0;0;2.13601;\\ %
        6.2;19200;0;0;1.97766;\\ %
        6.3;19200;0;0;1.83561;\\ %
        6.4;19200;0;0;1.6976;\\ %
        6.5;19200;0;0;1.55385;\\ %
        6.6;19200;0;0;1.42738;\\ %
        6.7;19200;0;0;1.30502;\\ %
        6.8;19200;0;0;1.18797;\\ %
        6.9;19200;0;0;1.07779;\\ %
        7;19200;0;0;0.974901;\\ %
    };%

\addplot [color=black, line width=0.7pt, mark=x, mark options={fill=white, mark size=1.5pt}]table[x=EbNo_dB,y=normalizedSumCalls, col sep=semicolon,row sep=crcr] {
EbNo_dB;simulatedFrames;FER;BER;normalizedSumCalls\\ 
3.4;12800;1;0.0293469;16.6304;\\ %
3.5;12800;1;0.0278473;16.5922;\\ %
3.6;12800;1;0.026349;16.5434;\\ %
3.7;12800;1;0.0248132;16.4742;\\ %
3.8;12800;1;0.0232651;16.3871;\\ %
3.9;12800;1;0.0216749;16.2614;\\ %
4;12800;1;0.0200528;16.0965;\\ %
4.1;12800;1;0.018385;15.8611;\\ %
4.2;12800;1;0.0166543;15.5309;\\ %
4.3;12800;1;0.0148062;15.054;\\ %
4.4;12800;1;0.0128794;14.3884;\\ %
4.5;12800;1;0.0108489;13.4291;\\ %
4.6;12800;0.999531;0.00860023;11.9919;\\ %
4.7;12800;0.957578;0.00593931;9.67064;\\ %
4.8;12800;0.605859;0.00258064;5.89693;\\ %
4.9;12800;0.129922;0.000412194;2.66862;\\ %
5;12800;0.00625;1.55342e-05;1.58365;\\ %
5.1;12800;0;0;1.27561;\\ %
5.2;12800;0;0;1.09637;\\ %
5.3;12800;0;0;0.965626;\\ %
5.4;12800;0;0;0.86324;\\ %
5.5;12800;0;0;0.778523;\\ %
5.6;12800;0;0;0.704419;\\ %
5.7;12800;0;0;0.638523;\\ %
5.8;12800;0;0;0.580639;\\ %
5.9;12800;0;0;0.52697;\\ %
6;12800;0;0;0.478463;\\ %
6.1;12800;0;0;0.434807;\\ %
6.2;12800;0;0;0.394392;\\ %
6.3;12800;0;0;0.357437;\\ %
6.4;12800;0;0;0.323845;\\ %
6.5;12800;0;0;0.293194;\\ %
6.6;12800;0;0;0.265595;\\ %
6.7;12800;0;0;0.240007;\\ %
6.8;12800;0;0;0.217223;\\ %
6.9;12800;0;0;0.195699;\\ %
7;12800;0;0;0.176339;\\ %
 }node [pos=0.4,anchor=south,font=\footnotesize,sloped] {iBDD};

\addplot[
    only marks,
    color=red,
    mark=o,
    mark size=3pt,
    mark options={fill=red!20}
]table[x=EbNo_dB,y=normalizedSumCalls, col sep=semicolon,row sep=crcr] {
EbNo_dB;simulatedFrames;FER;BER;normalizedSumCalls\\ 
4.5;12800;0;0;3.57783;\\ 
4.5;19200;0;0;5.09662;\\
4.5;19200;0;0;7.3808;\\
5.2;12800;0;0;1.09637;\\ %
};

\end{axis}

\end{tikzpicture}
    \caption{\ac{BDD} executions per CN in a PC with $[256,239,6]$ eBCH component code when decoding using iBDD and \acs{RDRSD} with different $\mathcal{J}$ values.}
    \label{fig:BDDruns}
\end{figure}
\subsubsection{Computational Complexity}

\begin{figure*}[!t]
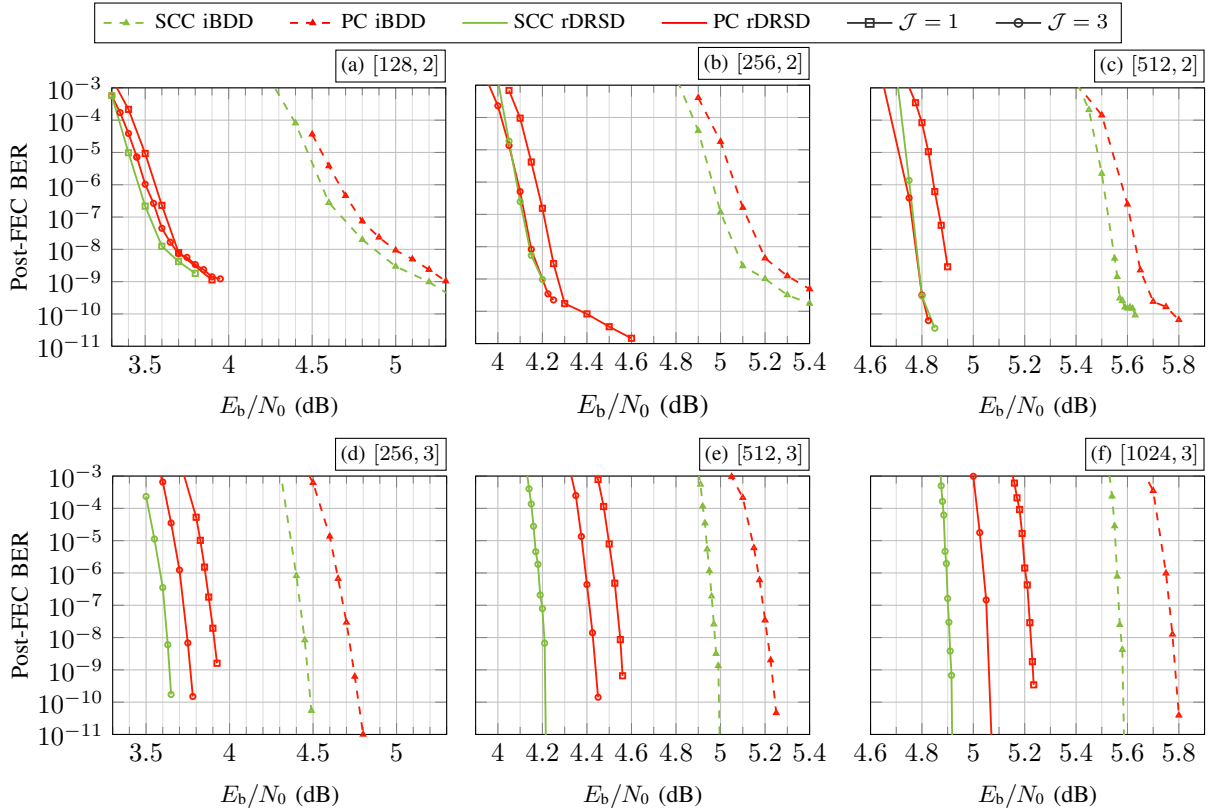

    \centering
    \setlength{\tabcolsep}{0pt}
    \begin{tabular}{@{}ccc@{}}
        \multicolumn{3}{c}{\begin{tikzpicture}
\begin{axis}[
    hide axis,
    scale only axis,
    width=0pt,
    height=0pt,
    xmin=0, xmax=1,
    ymin=0, ymax=1,
    legend columns=6,
    legend style={
        draw=black,
        at={(0.5,0)},
        anchor=center,
        /tikz/every even column/.append style={column sep=0.4cm},
        font = \footnotesize
    },
]

    \addlegendimage{color=KITpalegreen, dashed, line width=0.7pt, mark=triangle, mark options={ solid, mark size=1pt}}
    \addlegendentry{SCC iBDD}

    \addlegendimage{color=red, dashed, line width=0.7pt, mark=triangle, mark options={ solid, mark size=1pt}}
    \addlegendentry{PC iBDD}

        \addlegendimage{color=KITpalegreen, line width=0.7pt, mark=none, mark options={ solid, mark size=1pt}}
    \addlegendentry{SCC rDRSD}

    \addlegendimage{color=red, line width=0.7pt, mark=none, mark options={ solid, mark size=1pt}}
    \addlegendentry{PC rDRSD}
    
    \addlegendimage{color=black!80, line width=0.7pt, mark=square, mark options={ solid, mark size=1.5pt}}
    \addlegendentry{$\mathcal{J}=1$}
    
     \addlegendimage{color=black!80, line width=0.7pt, mark=o, mark options={ solid, mark size=1.5pt}}
    \addlegendentry{$\mathcal{J}=3$}

\end{axis}
\end{tikzpicture}} \\[-0.5ex]

        \begin{tikzpicture}

    \begin{axis}[%
        xshift=1.5cm,
        xmin=3.3,
        xmax=5.3,
        ymode=log,
        ymin=1e-11,
        ymax=0.001,
        yminorticks=true,
        axis background/.style={fill=white, mark size=1.5pt},
    xmajorgrids,
    xminorgrids,
    ymajorgrids,
    yminorgrids,
    width=6cm,
    height=5cm,
    xtick={3.0,3.5,...,10.5},
    minor x tick num=4,
    minor grid style={gray!25},%
 ytick={0.1,0.01,0.001,1e-4,1e-5,1e-6,1e-7,1e-8,1e-9,1e-10,1e-11,1e-12,1e-13,1e-14,1e-15},
    xlabel={$\Eb/\No$ (dB)},
    ylabel={Post-FEC BER},
    label style={font=\small},
    legend cell align={left},
    legend style={anchor = south west,at={(0,1.05)}, draw=black, fill opacity=0.7, text opacity = 1,legend columns=4,font=\footnotesize, row sep = 0pt}
        ]

        \addplot [color=orange!20!red, dashed, line width=0.7pt, mark=triangle, mark options={solid, mark size=1pt}]table[x=EbNo,y=BER, col sep=semicolon,row sep=crcr] {
        EbNo;EsNo;errorProb;noiseVar;ErasureProb;decodedFrame;FrameErr;FER;BER\\
            4.5;3.40829;0.0181402;0.228108;0;40140;318;0.00792227;3.64618e-05\\
            4.6;3.50829;0.0170868;0.222916;0;238650;301;0.00126126;3.80756e-06\\
            4.7;3.60829;0.016075;0.217842;0;1.12696e+06;301;0.00026709;4.57068e-07\\
            4.8;3.70829;0.015104;0.212883;0;3.66072e+06;301;8.22243e-05;7.46904e-08\\
            4.9;3.80829;0.0141736;0.208037;0;8.63055e+06;301;3.48761e-05;2.35485e-08\\
            5;3.90829;0.013283;0.203302;0;1.91484e+07;301;1.57193e-05;9.36071e-09\\
            5.1;4.00829;0.0124318;0.198674;0;3.60929e+07;301;8.33959e-06;4.9816e-09\\
            5.2;4.10829;0.0116191;0.194152;0;7.34748e+07;301;4.09664e-06;2.35006e-09\\
            5.3;4.21737;0.0107759;0.189336;0;1.65022e+08;301;1.824e-06;1.03635e-09\\
            5.4;4.31737;0.0100416;0.185026;0;2.31354e+08;187;8.08284e-07;4.47961e-10\\
            5.5;4.41737;0.00934371;0.180814;0;3.34416e+08;151;4.51533e-07;2.50772e-10\\
        };

        \addplot [color=KITpalegreen, dashed, line width=0.7pt, mark=triangle, mark options={ solid, mark size=1pt}]table[x=EbNo,y=BER, col sep=semicolon,row sep=crcr] {
            EbNo;        EsNo;       delta; ErasureProb;  totalFrame;    FE;         FER;         BER;      SPsize;          BE;  throughput;    mis rate\\
            4;     2.84016;   0.0249275;           0;         425;   425;           1;   0.0155009;     5079.33; 2.15872e+06;  5.1066e+08;           0\\%
            4.2;     3.04016;   0.0223808;           0;         429;   429;           1;  0.00466352;     1528.14;      655573; 2.84293e+08;           0\\%
            4.4;     3.24016;   0.0200056;           0;         679;   191;    0.281296; 8.08019e-05;     94.1257;       17978; 1.21346e+08;           0\\%
            4.6;     3.44016;   0.0177999;           0;        8400;    42;       0.005; 2.78291e-07;     18.2381;         766; 1.02628e+08;           0\\%
            4.8;     3.64016;   0.0157611;           0;       58800;    40; 0.000680272; 1.99817e-08;       9.625;         385;  9.1298e+07;           0\\%
            5;     3.84016;   0.0138854;           0;      384000;    40; 0.000104167; 2.95639e-09;         9.3;         372; 8.87341e+07;           0\\%
            5.2;     4.04016;   0.0121686;           0;      982661;    34; 3.45999e-05; 9.78266e-10;     9.26471;         315; 9.33949e+07;           0\\%
            5.4;     4.24016;   0.0106053;           0; 2.78706e+06;    22; 7.89362e-06; 2.16805e-10;           9;         198; 1.07179e+08;           0\\%
            5.6;     4.44016;  0.00918963;           0; 2.84623e+06;     5; 1.75671e-06; 4.82495e-11;           9;          45; 1.18366e+08;           0\\%
        };

   \addplot [color=orange!20!red, line width=0.7pt, mark=o, mark options={ solid, mark size=1pt}]table[x=EbNo,y=BER, col sep=semicolon,row sep=crcr] {
            EbNo;        EsNo;       delta; ErasureProb;  totalFrame;    FE;         FER;         BER;      SPsize;          BE;  throughput;    mis rate\\
            3.3;     2.21737;   0.0163461;   0.0485185;         840;    58;   0.0690476; 0.000565156;     134.103;        7778; 9.65403e+07;  0.00158539\\%
            3.35;     2.26737;     0.01585;   0.0479146;        2160;    51;   0.0236111; 0.000171619;     119.088;      6073.5; 2.23521e+08;  0.00107182\\%
            3.4;     2.31737;   0.0153641;   0.0473089;        7320;    52;  0.00710383; 3.88349e-05;     89.5673;      4657.5; 1.66641e+08; 0.000782025\\%
            3.45;     2.36737;   0.0148883;   0.0467014;       29520;    50;  0.00169377; 7.18486e-06;        69.5;        3475; 1.67929e+08; 0.000606194\\%
            3.45;     2.36737;   0.0148883;   0.0467014;       29520;    50;  0.00169377; 7.18486e-06;        69.5;        3475; 1.67929e+08; 0.000606194\\%
            3.5;     2.41737;   0.0144227;   0.0460925;      130920;    50; 0.000381913;  1.0373e-06;        44.5;        2225; 1.86105e+08; 0.000509687\\%
            3.55;     2.46737;    0.013967;   0.0454821;      417360;    50; 0.000119801; 2.66086e-07;       36.39;      1819.5; 1.89693e+08; 0.000428549\\%
            3.6;     2.51737;   0.0135213;   0.0448706; 1.00646e+06;    50; 4.96789e-05; 4.48153e-08;       14.78;         739; 1.89663e+08; 0.000368085\\%
            3.65;     2.56737;   0.0130855;    0.044258;  2.8378e+06;    50; 1.76193e-05; 1.66686e-08;        15.5;         775; 1.91227e+08; 0.000321787\\%
            3.7;     2.61737;   0.0126594;   0.0436446; 3.86573e+06;    50; 1.29342e-05; 7.29441e-09;        9.24;         462; 1.93416e+08; 0.000283873\\%
            3.75;     2.66737;   0.0122431;   0.0430306; 5.35647e+06;    50;  9.3345e-06; 5.64036e-09;         9.9;         495; 1.96009e+08; 0.000254516\\%
            3.8;     2.71737;   0.0118364;    0.042416; 8.99083e+06;    50; 5.56122e-06; 3.38072e-09;        9.96;         498; 1.98005e+08; 0.000230759\\%
            3.85;     2.76737;   0.0114392;   0.0418012; 1.19664e+07;    50; 4.17836e-06; 2.35645e-09;        9.24;         462; 1.99263e+08; 0.000209566\\%
            3.9;     2.81737;   0.0110514;   0.0411862; 2.04557e+07;    50; 2.44431e-06; 1.41431e-09;        9.48;         474; 2.00994e+08;  0.00019233\\%
            3.95;     2.86737;    0.010673;   0.0405713; 2.31048e+07;    50; 2.16405e-06; 1.22045e-09;        9.24;         462; 2.02347e+08; 0.000177206\\%
        };

        \addplot [color=orange!20!red, line width=0.7pt, mark=square, mark options={ solid, mark size=1pt}]table[x=EbNo,y=BER, col sep=semicolon,row sep=crcr] {
  EbNo;        EsNo;       delta; ErasureProb;  totalFrame;    FE;         FER;         BER;      SPsize;          BE;  throughput;    mis rate\\
   3.1;     2.01737;   0.0192013;   0.0479464;       19200; 17990;    0.936979;   0.0167113;     292.213; 5.25692e+06; 2.43731e+08;  0.00379894\\ %
   3.2;     2.11737;   0.0181074;   0.0468149;       19200; 13078;    0.681146;  0.00771644;     185.608; 2.42738e+06; 2.79907e+08;  0.00244778\\ %
   3.3;     2.21737;   0.0170554;   0.0456745;       19200;  5350;    0.278646;  0.00195537;     114.973;      615108; 3.71087e+08;  0.00138815\\ %
   3.4;     2.31737;   0.0160448;   0.0445264;       19200;   931;   0.0484896; 0.000215694;     72.8802;     67851.5; 4.65616e+08; 0.000788755\\ %
   3.5;     2.41737;   0.0150751;   0.0433719;       76800;   250;  0.00325521; 9.26097e-06;      46.612;       11653; 4.73012e+08; 0.000497705\\ %
   3.6;     2.51737;   0.0141459;   0.0422124;   1.536e+06;   203; 0.000132161; 2.30392e-07;     28.5616;        5798; 4.76964e+08; 0.000352214\\ %
   3.7;     2.61737;   0.0132565;   0.0410491; 1.66464e+07;   200; 1.20146e-05; 7.92895e-09;     10.8125;      2162.5; 4.80378e+08; 0.000260859\\ %
   3.9;     2.81737;   0.0115949;   0.0387167; 9.89184e+07;   200; 2.02187e-06; 1.13471e-09;       9.195;        1839; 4.95328e+08; 0.000164959\\ %
        };

        \addplot [color=KITpalegreen, line width=0.7pt, mark=square, mark options={ solid, mark size=1pt}]table[x=EbNo,y=BER, col sep=semicolon,row sep=crcr] {
  EbNo;        EsNo;       delta; ErasureProb;  totalFrame;    FE;         FER;         BER;      SPsize;          BE;  throughput;    mis rate\\
   3.1;     1.94016;   0.0176805;   0.0583726;       19200; 19200;           1;   0.0276403;     9057.18; 1.73898e+08; 2.61518e+08;           0\\ %
   3.2;     2.04016;   0.0166451;   0.0570635;       19200; 19126;    0.996146;   0.0114083;     3752.73; 7.17747e+07; 2.91395e+08;           0\\ %
   3.3;     2.14016;   0.0156509;   0.0557427;       19200; 10551;    0.549531; 0.000578761;      345.11; 3.64125e+06; 3.46363e+08;           0\\ %
   3.4;     2.24016;   0.0146976;   0.0544115;       19200;   901;   0.0469271; 9.78843e-06;     68.3502;     61583.5; 3.72155e+08;           0\\ %
   3.5;     2.34016;   0.0137844;   0.0530713;       96000;   245;  0.00255208; 2.17724e-07;     27.9551;        6849; 3.79921e+08;           0\\ %
   3.6;     2.44016;   0.0129109;   0.0517237;      576000;   203; 0.000352431; 1.24878e-08;     11.6108;        2357; 3.86677e+08;           0\\ %
   3.7;     2.54016;   0.0120763;   0.0503701;  1.4592e+06;   204; 0.000139803; 4.20265e-09;     9.85049;      2009.5; 3.93727e+08;           0\\ %
     3.8;     2.64016;   0.0112802;    0.049012;  3.2256e+06;   200;  6.2004e-05; 1.80895e-09;        9.56;        1912; 4.04733e+08;           0\\ %
        };

    \end{axis}

\node[
    anchor=south east,
    font=\footnotesize,
    draw=black,
    fill=white,
    inner sep=2pt,
    yshift=2pt
] at (current axis.north east) {(a) $[128,2]$};

\end{tikzpicture} &
        \begin{tikzpicture}

    \begin{axis}[%
        xshift=1.5cm,
        xmin=3.9,
        xmax=5.4,
        ymode=log,
        ymin=1e-11,
        ymax=0.001,
        yminorticks=true,
        axis background/.style={fill=white, mark size=1.5pt},
    xmajorgrids,
    xminorgrids,
    ymajorgrids,
    yminorgrids,
    width=6cm,
    height=5cm,
    xtick={3.0,3.2,...,10.6},
    minor x tick num=1,
    minor grid style={gray!25},%
ytick={0.1,0.01,0.001,1e-4,1e-5,1e-6,1e-7,1e-8,1e-9,1e-10,1e-11,1e-12,1e-13,1e-14,1e-15},
    xlabel={$\Eb/\No$ (dB)},
    ylabel style={white},
    legend cell align={left},
    legend style={anchor = south west,at={(0,1.05)}, draw=black, fill opacity=1, text opacity = 1,legend columns=3,font=\footnotesize, row sep = 0pt},
    yticklabels=\empty
        ]

        \addplot [color=orange!20!red, dashed, line width=0.7pt, mark=triangle, mark options={solid,mark size=1pt}]table[x=EbNo,y=BER, col sep=semicolon,row sep=crcr] {
        EbNo;EsNo;errorProb;noiseVar;ErasureProb;decodedFrame;FrameErr;FER;BER\\
            4.9;4.30074;0.0101612;0.185736;0;1221;229;0.187551;0.000409091\\
            5;4.40074;0.00945732;0.181508;0;14785;202;0.0136625;1.78948e-05\\
            5.1;4.50074;0.00878902;0.177377;0;3.42077e+06;501;0.000146458;1.61791e-07\\
            5.2;4.60074;0.00815547;0.173339;0;8.68736e+06;201;2.31371e-05;4.39727e-09\\
            5.3;4.70074;0.00755577;0.169393;0;2.37323e+07;201;8.46946e-06;1.24093e-09\\
            5.4;4.80074;0.00698897;0.165538;0;2.3e+07;79;3.43478e-06;4.89444e-10\\
            5.5;4.90074;0.00645414;0.161769;0;1.56355e+08;240;1.53497e-06;2.17468e-10\\
            5.6;5.00074;0.00595027;0.158087;0;4.29515e+08;290;6.75181e-07;9.4847e-11\\
        };

        \addplot [color=KITpalegreen, dashed, line width=0.7pt, mark=triangle, mark options={solid,mark size=1pt}]table[x=EbNo,y=BER, col sep=semicolon,row sep=crcr] {
            EbNo;        EsNo;       delta; ErasureProb;  totalFrame;    FE;         FER;         BER;      SPsize;          BE;  throughput;    mis rate\\
            4.7;     4.08113;    0.011836;           0;         417;   417;           1;  0.00585435;     7673.42; 3.19982e+06; 7.60264e+08;           0\\%
            4.8;     4.18113;   0.0110511;           0;         422;   421;     0.99763;  0.00177283;      2329.2;      980593; 3.03809e+08;           0\\%
            4.9;     4.28113;   0.0103035;           0;         838;   234;    0.279236; 3.87688e-05;     181.979;       42583; 1.29647e+08;           0\\%
            5;     4.38113;  0.00959249;           0;       10439;    43;  0.00411917; 1.16571e-07;      37.093;        1595; 1.05979e+08;           0\\%
            5.1;     4.48113;  0.00891728;           0;      112800;    40;  0.00035461; 2.52284e-09;       9.325;         373; 1.00274e+08;           0\\%
            5.2;     4.58113;  0.00827699;           0;      283200;    40; 0.000141243; 1.01294e-09;         9.4;         376; 9.74792e+07;           0\\%
            5.3;     4.68113;  0.00767072;           0;      559285;    25; 4.46999e-05; 3.19207e-10;        9.36;         234; 9.47855e+07;           0\\%
            5.4;     4.78113;  0.00709755;           0;      752819;    17; 2.25818e-05; 1.73299e-10;     10.0588;         171; 9.24872e+07;           0\\%
        };

        \addplot [color=orange!20!red, line width=0.7pt, mark=square, mark options={solid,mark size=1pt}]table[x=EbNo,y=BER, col sep=semicolon,row sep=crcr] {%
            EbNo;        EsNo;       delta; ErasureProb;  totalFrame;    FE;         FER;         BER;      SPsize;          BE\\
            4.05;     3.45316;  0.00869831;   0.0248508;         307;    93;    0.302932; 0.000687143;     148.656;       13825\\
            4.1;     3.50316;  0.00837965;   0.0243882;         763;    42;   0.0550459; 9.54424e-05;     113.631;      4772.5\\
            4.15;     3.55316;  0.00806954;   0.0239277;        3337;    32;  0.00958945; 4.18394e-06;     28.5938;         915\\
            4.2;     3.60316;  0.00776787;   0.0234695;       83493;    31; 0.000371289; 1.54702e-07;     27.3065;       846.5\\
            4.25;     3.65316;   0.0074745;   0.0230137; 2.31514e+06;    31; 1.33901e-05; 2.99885e-09;     14.6774;         455\\
            4.3;     3.70316;  0.00718933;   0.0225605; 3.04622e+07;    31; 1.01766e-06; 1.73816e-10;     11.1935;         347\\
            4.4;     3.80316;   0.0066431;   0.0216623;   5.194e+07;    31; 5.96842e-07; 8.28452e-11;     9.09677;         282\\
            4.5;     3.90316;   0.0061282;   0.0207759; 1.27926e+08;    31; 2.42328e-07; 3.32787e-11;           9;         279\\
            4.6;     4.00316;  0.00564362;   0.0199023;     9.4e+06;     1; 1.06383e-07; 1.46095e-11;           9;           9\\ %
        };

        \addplot [color=orange!20!red, line width=0.7pt, mark=o, mark options={solid,mark size=1pt}]table[x=EbNo,y=BER, col sep=semicolon,row sep=crcr] {%
  EbNo;        EsNo;       delta; ErasureProb;  totalFrame;    FE;         FER;         BER;      SPsize;          BE;  throughput;    mis rate\\
   3.7;     3.10316;   0.0106194;   0.0303733;       19200; 19200;           1;   0.0204586;     1340.77; 2.57428e+07; 4.04424e+08;  0.00550635\\ %
  3.75;     3.15316;   0.0102515;   0.0298662;       19200; 19198;    0.999896;   0.0181008;     1186.38; 2.27761e+07;  4.0134e+08;  0.00539567\\ %
   3.8;     3.20316;  0.00989283;   0.0293603;       19200; 19080;     0.99375;   0.0148233;     977.571; 1.86521e+07;  3.9447e+08;  0.00497045\\ %
  3.85;     3.25316;  0.00954314;   0.0288557;       19200; 17897;    0.932135;   0.0102068;     717.616; 1.28432e+07; 4.00411e+08;   0.0039242\\ %
   3.9;     3.30316;  0.00920238;   0.0283525;       19200; 13326;    0.694063;  0.00501819;     473.836; 6.31434e+06; 4.17537e+08;   0.0024268\\ %
  3.95;     3.35316;  0.00887045;    0.027851;       19200;  5929;    0.308802;  0.00146521;     310.957; 1.84366e+06; 4.43641e+08;  0.00120002\\ %
     4;     3.40316;  0.00854723;   0.0273513;       19200;  1365;   0.0710937; 0.000230654;     212.623;      290230; 4.61547e+08; 0.000603549\\ %
  4.05;     3.45316;  0.00823261;   0.0268536;       38400;   226;  0.00588542; 1.36391e-05;     151.876;       34324; 4.75113e+08; 0.000343469\\ %
   4.1;     3.50316;  0.00792649;    0.026358;      768000;   202; 0.000263021; 5.09342e-07;     126.911;       25636; 4.92798e+08; 0.000239754\\ %
  4.15;     3.55316;  0.00762874;   0.0258646; 3.45408e+07;   200; 5.79025e-06; 8.27861e-09;        93.7;       18740; 5.02583e+08; 0.000180059\\ %
      4.225;     3.62816;  0.00719756;    0.025129; 8.13696e+07;   200; 2.45792e-06; 3.46545e-10;        9.24;        1848; 5.15572e+08; 0.000125514\\ %
     4.25;     3.65316;  0.00705789;   0.0248851; 4.99392e+07;    73; 1.46178e-06; 2.24425e-10;     10.0616;       734.5; 5.23802e+08; 0.000113084\\ %
};

        \addplot [color=KITpalegreen, line width=0.7pt, mark=o, mark options={mark size=1pt}]table[x=EbNo,y=BER, col sep=semicolon,row sep=crcr] {
  EbNo;        EsNo;       delta; ErasureProb;  totalFrame;    FE;         FER;         BER;      SPsize;          BE;  throughput;    mis rate\\
   3.7;     3.08113;  0.00927709;   0.0373591;        3840;  3840;           1;   0.0252435;     33087.2; 1.27055e+08; 3.78324e+08;           0\\ %
  3.75;     3.13113;  0.00894321;   0.0367555;        3840;  3840;           1;   0.0239401;     31378.8; 1.20495e+08; 3.78313e+08;           0\\ %
   3.8;     3.18113;  0.00861807;    0.036153;        3840;  3840;           1;     0.02247;     29451.9; 1.13095e+08; 3.66487e+08;           0\\ %
  3.85;     3.23113;  0.00830156;    0.035552;        3840;  3840;           1;   0.0206105;     27014.5; 1.03736e+08; 3.61381e+08;           0\\ %
   3.9;     3.28113;  0.00799356;   0.0349525;        3840;  3840;           1;   0.0173991;     22805.3; 8.75723e+07; 3.47605e+08;           0\\ %
  3.95;     3.33113;  0.00769396;   0.0343547;        3840;  3708;    0.965625;   0.0110883;       15051; 5.58092e+07; 3.44457e+08;           0\\ %
     4;     3.38113;  0.00740265;   0.0337588;        3840;  1509;    0.392969;  0.00141301;        4713; 7.11192e+06; 3.84268e+08;           0\\ %
  4.05;     3.43113;  0.00711951;    0.033165;        7680;   363;   0.0472656; 1.85721e-05;     515.021;      186952; 4.08448e+08;           0\\ %
   4.1;     3.48113;  0.00684441;   0.0325734;       61440;   212;  0.00345052; 2.47514e-07;     94.0212;     19932.5; 4.01796e+08;           0\\ %
  4.15;     3.53113;  0.00657724;   0.0319843;      602880;   202; 0.000335058; 5.38024e-09;      21.047;      4251.5;  4.2163e+08;           0\\ %
   4.2;     3.58113;  0.00631786;   0.0313978;  1.4592e+06;   200; 0.000137061; 9.64654e-10;       9.225;        1845; 4.30108e+08;           0\\ %
        };

    \end{axis}

\node[
    anchor=south east,
    font=\footnotesize,
    draw=black,
    fill=white,
    inner sep=2pt,
    yshift=2pt
] at (current axis.north east) {(b) $[256,2]$};

\end{tikzpicture} &
        \input{numerical/T2_511} \\
        \begin{tikzpicture}

    \begin{axis}[%
        xshift=1.5cm,
        xmin=3.3,
        xmax=5.3,
        ymode=log,
        ymin=1e-11,
        ymax=0.001,
        yminorticks=true,
        axis background/.style={fill=white, mark size=1.5pt},
    xmajorgrids,
    xminorgrids,
    ymajorgrids,
    yminorgrids,
    width=6cm,
    height=5cm,
        xtick={3.0,3.5,...,10.5},
    minor x tick num=4,
    minor grid style={gray!25},%
 ytick={0.1,0.01,0.001,1e-4,1e-5,1e-6,1e-7,1e-8,1e-9,1e-10,1e-11,1e-12,1e-13,1e-14,1e-15},
    xlabel={$\Eb/\No$ (dB)},
    ylabel={Post-FEC BER},
    label style={font=\small},
    legend cell align={left},
    legend style={anchor = south west,at={(0,1.05)}, draw=black, fill opacity=0.7, text opacity = 1,legend columns=4,font=\footnotesize, row sep = 0pt},
        ]

        \addplot [color=orange!20!red, dashed, line width=0.7pt, mark=triangle, mark options={solid, mark size=1pt}]table[x=EbNo,y=BER, col sep=semicolon,row sep=crcr] {
            EbNo;EsNo; delta; ErasureProb;totalFrame;FE; FER; BER;SPsize;BE;throughput;mis rate\\
            3.8; 2.94144; 0.0236164; 0; 355; 355; 1; 0.0246849; 1605.14;569823;6.6038e+09; 0\\
            4; 3.14144; 0.0211567; 0; 355; 355; 1; 0.0203842; 1325.48;470546; 6.39672e+09; 0\\
            4.2; 3.34144; 0.0188676; 0; 355; 355; 1; 0.0149838; 974.324;345885;7.1296e+09; 0\\
            4.4; 3.54144; 0.0167468; 0; 515; 292; 0.56699;0.00509719; 584.568;170694; 8.67444e+09; 0\\
            4.5; 3.64144; 0.0157486; 0;1573; 137; 0.0870947; 0.000621386; 463.927; 63558;1.5516e+10; 0\\
            4.6; 3.74144; 0.0147911; 0; 48779; 102;0.00209106; 1.35634e-05; 421.775; 43021; 2.42567e+10; 0\\
            4.65; 3.79144; 0.0143276; 0;882183; 101; 0.000114489; 6.68503e-07; 379.683; 38348; 2.53622e+10; 0\\
            4.7; 3.84144;0.013874; 0;1.9404e+07; 100; 5.15358e-06; 2.94941e-08;372.14; 37214; 2.58426e+10; 0\\
            4.75; 3.89144; 0.0134303; 0; 8.86229e+08; 100; 1.12838e-07; 6.23249e-10;359.16; 35916; 2.61823e+10; 0\\
            4.8; 3.94144; 0.0129965; 0;2.13e+09; 4; 1.87793e-09; 1.00286e-11;347.25;1389; 2.648e+10; 0\\ %
        };

        \addplot [color=KITpalegreen, dashed, line width=0.7pt, mark=triangle, mark options={ solid, mark size=1pt}]table[x=EbNo,y=BER, col sep=semicolon,row sep=crcr] {
            EbNo;        EsNo;       delta; ErasureProb;  totalFrame;    FE;         FER;         BER;      SPsize;          BE;  throughput;    mis rate\\
            4.2;     3.25627;   0.0198216;           0;       12800; 12800;           1;   0.0122669;     16078.4; 2.05804e+08; 2.97892e+09;           0\\ %
            4.3;     3.35627;   0.0187046;           0;       12800;  8916;    0.696562;  0.00187548;     3529.08; 3.14653e+07; 6.52312e+08;           0\\ %
            4.4;     3.45627;   0.0176295;           0;       12800;    50;  0.00390625; 8.15094e-07;       273.5;       13675; 4.55931e+08;           0\\ %
            4.45;     3.50627;   0.0171076;           0;      532800;    40; 7.50751e-05; 8.37401e-09;       146.2;        5848;  3.1695e+08;           0\\ %
            4.49;     3.54627;   0.0166976;           0;  2.6002e+06;     6; 2.30751e-06; 5.45753e-11;          31;         186; 3.15552e+08;           0\\ %
        };

       \addplot [color=orange!20!red, line width=0.7pt, mark=square, mark options={mark size=1pt}]table[x=EbNo,y=BER, col sep=semicolon,row sep=crcr] {
  EbNo;        EsNo;       delta; ErasureProb;  totalFrame;    FE;         FER;         BER;      SPsize;          BE;  throughput;    mis rate\\
   3.3;     2.40744;   0.0207803;   0.0243624;        3840;  3840;           1;   0.0292895;     1919.51; 7.37093e+06;  2.8893e+08;  0.00178702\\ %
   3.4;     2.50744;   0.0196281;   0.0236814;        3840;  3840;           1;   0.0260622;     1708.02; 6.55878e+06; 3.01956e+08;   0.0016195\\ %
   3.5;     2.60744;   0.0185182;   0.0229985;        3840;  3840;           1;   0.0217421;     1424.89; 5.47159e+06;  3.9742e+08;  0.00131326\\ %
   3.6;     2.70744;   0.0174503;   0.0223146;        3840;  3724;    0.969792;   0.0149566;     1010.73; 3.76394e+06; 4.36737e+08; 0.000873657\\ %
   3.7;     2.80744;    0.016424;   0.0216303;        3840;  1494;    0.389062;  0.00325361;     548.058;      818798; 4.66417e+08; 0.000366426\\ %
   3.8;     2.90744;   0.0154388;   0.0209467;       23040;   227;  0.00985243; 5.30415e-05;     352.819;       80090; 4.91596e+08;  0.00013353\\ %
 3.825;     2.93244;   0.0151989;    0.020776;       76800;   206;  0.00268229; 1.02234e-05;     249.786;       51456; 4.96305e+08; 0.000104036\\ %
  3.85;     2.95744;   0.0149615;   0.0206054;      441600;   202; 0.000457428; 1.51935e-06;     217.678;       43971; 5.21924e+08; 8.64084e-05\\ %
 3.875;     2.98244;   0.0147266;   0.0204349;  3.4176e+06;   200; 5.85206e-05; 1.79318e-07;     200.815;       40163; 4.77215e+08; 7.01555e-05\\ %
   3.9;     3.00744;   0.0144943;   0.0202646; 2.56128e+07;   200;  7.8086e-06; 1.94673e-08;     163.385;       32677; 5.26845e+08; 6.01738e-05\\ %
    3.925;     3.03244;   0.0142644;   0.0200944;  8.6784e+07;    45; 5.18529e-07; 1.60326e-09;     202.633;      9118.5; 5.26546e+08; 5.07481e-05\\ %
       };
        \addplot [color=orange!20!red, line width=0.7pt, mark=o, mark options={ solid, mark size=1pt}]table[x=EbNo,y=BER, col sep=semicolon,row sep=crcr] {
            EbNo;        EsNo;       delta; ErasureProb;  totalFrame;    FE;         FER;         BER;      SPsize;          BE;  throughput;    mis rate\\
            3.4;     2.54144;   0.0161274;    0.033491;         446;   445;    0.997758;   0.0276424;     1801.49;      801662; 1.71936e+08;   0.0186278\\%
            3.45;     2.59144;   0.0156357;   0.0330147;         468;   439;    0.938034;   0.0219195;     1519.47;      667046; 2.14294e+08;   0.0161414\\%
            3.5;     2.64144;   0.0151543;   0.0325379;         618;   445;    0.720065;   0.0123248;     1112.98;      495276; 2.25144e+08;   0.0101318\\%
            3.55;     2.69144;    0.014683;   0.0320607;         928;   291;    0.313578;  0.00404454;     838.696;      244060; 2.63798e+08;  0.00425879\\%
            3.6;     2.74144;   0.0142217;   0.0315833;        1200;    77;   0.0641667; 0.000653274;     662.013;       50975; 2.97856e+08;  0.00136025\\%
            3.65;     2.79144;   0.0137704;   0.0311059;       15600;    56;  0.00358974; 3.52384e-05;     638.312;     35745.5; 3.11185e+08; 0.000539083\\%
            3.7;     2.84144;   0.0133291;   0.0306285;      328800;    50; 0.000152068; 1.23088e-06;      526.33;     26316.5; 3.19828e+08; 0.000351809\\%
            3.75;     2.89144;   0.0128975;   0.0301514; 4.50995e+07;    50; 1.10866e-06; 6.81734e-09;      399.85;     19992.5; 1.06303e+09; 0.000264272\\ %
            3.78;     2.92144;   0.0126433;   0.0298653; 3.23251e+08;    12; 3.71228e-08; 1.50171e-10;     263.042;      3156.5; 1.06414e+09; 0.000231599\\
        };

        \addplot [color=KITpalegreen, line width=0.7pt, mark=o, mark options={ solid, mark size=1pt}]table[x=EbNo,y=BER, col sep=semicolon,row sep=crcr] {
            EbNo;        EsNo;       delta; ErasureProb;  totalFrame;    FE;         FER;         BER;      SPsize;          BE;  throughput;    mis rate\\
            3.5;     2.55627;   0.0128632;    0.045815;       12800;   153;   0.0119531; 0.000234243;     25685.9; 3.92994e+06; 1.03259e+09;           0\\ %
            3.55;     2.60627;   0.0124422;   0.0451852;       89600;    51; 0.000569196; 1.12788e-05;     25972.3; 1.32459e+06; 1.09348e+09;           0\\ %
            3.6;     2.65627;   0.0120309;   0.0445546;  1.0496e+06;    20; 1.90549e-05; 3.52108e-07;     24220.3;      484406; 1.03564e+09;           0\\ %
            3.63;     2.68627;   0.0117886;    0.044176;  6.8325e+06;     2; 2.92719e-07; 6.01229e-09;     26921.5;       53843; 6.22843e+08;           0\\ %
            3.65;     2.70627;   0.0116291;   0.0439235;  2.0625e+07;     2; 9.69697e-08;  1.7321e-10;     2341.25;      4682.5; 6.25558e+08;           0\\ %
        };

    \end{axis}

\node[
    anchor=south east,
    font=\footnotesize,
    draw=black,
    fill=white,
    inner sep=2pt,
    yshift=2pt
] at (current axis.north east) {(d) $[256,3]$};

\end{tikzpicture} &
        \input{numerical/T3_511} &
        \input{numerical/T3_1023} \\
    \end{tabular}
    \caption{Post-FEC BER performance of the \acs{RDRSD} for \acp{GPC} based on eBCH component codes with different component-code lengths \(n\) and error-correction capabilities~\(t\). The entries of \(\bm{D}\) and \(\bm{D}_{\text{CN}}\) use \(q=4\) bits. Dashed and solid curves correspond to \ac{iBDD} and \acs{RDRSD}, respectively, colors distinguish \acp{PC} and staircase codes, and markers indicate \(\mathcal{J}\).}
    \label{fig:T2codes}
\end{figure*}

The initialization complexity of the \acs{RDRSD} is $\mathcal{O}(N)$, as in iBDD. During iterative decoding, compared with \ac{iBDD}, the additional operations are limited to simple integer comparisons, additions, increments/decrements, and clipping operations on \(q\)-bit entries of \(\bm{D}\) and \(\bm{D}_{\text{CN}}\).
In particular, the \ac{DRS}-based miscorrection detection only requires checking whether the \(D_i\) value of any flipped position exceeds \(\ta\), and optionally accumulating the \(D_i\) values over the flipped positions and comparing the result with \(\tee\).
Since the number of flipped positions is bounded by the correction capability $t$ of the component code, this additional complexity is small compared with algebraic component decoding. Therefore, the computational complexity of the \acs{RDRSD} is dominated by the component-code \ac{EaED} operations, similarly to syndrome-domain \ac{iBDD}.

In Fig.~\ref{fig:BDDruns}, we compare the number of \ac{BDD} executions per CN for the \acs{RDRSD} and \ac{iBDD} on the $[256,239,6]$ eBCH-code-based PC across different $\Eb/\No$ values. In the waterfall region, \ac{iBDD} requires approximately $1.3$ to $16.6$ \ac{BDD} executions per CN, while at high \ac{SNR}, most frames are decoded within a few iterations and the number of \ac{BDD} executions per CN approaches $1$. In contrast, the \acs{RDRSD} performs more \ac{BDD} executions in the low- and high-\ac{SNR} regions due to its list-based \ac{EaED} operations, but fewer in the moderate-\ac{SNR} region because it can terminate earlier than \ac{iBDD}, which often uses all $20$ iterations after failing to converge. A practical comparison should therefore focus on the operating regions where the post-decoding
\ac{BER} of each decoder is low. For example, Fig.~\ref{fig:T2codes}\nobreakdash-(b) shows that the target post-decoding \ac{BER} of $10^{-9}$ is reached above approximately $\SI{5.2}{dB}$ for \ac{iBDD} and $\SI{4.5}{dB}$ for the \acs{RDRSD}. At high \ac{SNR}, e.g., for $\Eb/\No\geq\SI{4.8}{dB}$, the \acs{RDRSD} requires roughly $\mathcal{J}$ times as many \ac{BDD} executions as \ac{iBDD}, which is still below the worst-case list-size estimate because most erasures are resolved in the first few iterations and later \acs{RDRSD} iterations mainly process CNs with $E=0$, requiring only one \ac{BDD} step.

\subsubsection{Memory}
The \acs{RDRSD} requires a ternary data buffer of $2N$ bits, a syndrome buffer of $btM$ bits, a bit-\ac{DRS} buffer storing \(\bm{D}\) with $qN$ bits, and a CN-\ac{DRS} buffer storing \(\bm{D}_{\text{CN}}\) with $qM$ bits.
For a block \ac{GPC} with \(N\) VNs, \(M\) CNs, and component BCH code with error-correcting capability $t$, the total memory is approximately
\(
    M_{\text{rDRSD}}\approx (q+2)N + (q+bt)M  \quad \text{bits}.
\)

\subsubsection{Internal Decoder Data Flow}

The \acs{RDRSD} also preserves the low-data-flow advantage of syndrome-domain \ac{iBDD}. After initialization, it operates only on the ternary data buffer, the syndrome buffers, and the quantized \(q\)-bit entries of \(\bm{D}\) and \(\bm{D}_{\text{CN}}\). Thus, unlike \ac{SDD} algorithms, which pass multi-bit soft messages along all graph edges in every iteration, the \acs{RDRSD} yields only a limited additional internal data flow. During each \ac{CN} decoding step, this consists of reading and updating the \(bt\)-bit syndrome, reading the $E$ erased positions, updating the \(E\) erasures with binary values after the \ac{EaED} step, flipping up to \(t\) additional data bits, and sending decrease signals for up to \(2\mathcal{J}t\) entries of \(\bm{D}\).

\section{Numerical Results of the Waterfall Region}\label{sec:waterfallResults}

This section presents post-decoding \ac{BER} results for selected \acp{PC} and staircase codes with various component codes.
The decoder parameters are optimized using Optuna~\cite{optuna_2019} with the tree-structured Parzen estimator (TPE) sampler, a sequential model-based optimization method.
For each considered code and decoder configuration, we define a search space for \(\te\), $d_0$, \(\ta\), \(\tee\), and their iteration-dependent update schedules.
Each Optuna trial corresponds to one parameter configuration of the proposed \acs{RDRSD}.
The decoder is then instantiated and evaluated by Monte-Carlo simulation at representative \(\Eb/\No\) values in the waterfall region.
The resulting post-\ac{FEC} \ac{BER} is used as the objective value to be minimized.
Based on completed trials, Optuna builds probabilistic surrogate models of promising and non-promising parameter regions and samples new candidates preferentially from the promising regions.

Fig.~\ref{fig:T2codes} shows the post-\ac{FEC} \ac{BER} performance of \acp{GPC} based on eBCH component codes with \(t=2\) and \(t=3\). Tab.~\ref{tab:compare} further summarizes, for \acp{PC}, the decoding gain in \(\Eb/\No\) at a target \ac{BER} of \(10^{-8}\) relative to the \ac{iBDD} algorithm. For the PC comparisons in Tab.~\ref{tab:compare}, all decoders use 20 iterations. The \acs{RDRSD} uses \(q=4\)-bit entries for \(\bm{D}\) and \(\bm{D}_{\text{CN}}\), whereas conventional \ac{DRSD} uses \(q=5\). Despite the lower-precision reliability scores and reduced internal data flow, the \acs{RDRSD} achieves larger gains than conventional \ac{DRSD}. For \acp{PC} with large block length \(n\) and \(t=3\), the list decoding option gives a particularly noticeable improvement.
Similar trends are observed for staircase codes. To keep the figure readable, we therefore show only \acs{RDRSD} with \(\mathcal{J}=3\), except for the \([128,2]\)-based staircase code, where \(\mathcal{J}=1\) and \(\mathcal{J}=3\) perform similarly.

For $t=2$ codes, the \acs{RDRSD} achieves a small reduction in the error floor compared to \ac{iBDD}, but the resulting floors may still be too high for some \ac{FEC} applications. For codes with $t=3$, no error floors are observed for either \ac{iBDD} or the \acs{RDRSD} within the simulated range. Estimates in~\cite{staircaseCode} indicate that the error floors of such GPCs are below $10^{-20}$ and thus negligible for most practical applications. Since many systems still use $t=2$ component codes due to their low decoding complexity, we study their error-floor behavior in detail and propose new post-processing techniques to lower the error floors in the next sections.

\begin{table}[!t]
    \caption{Decoding gain with respect to \ac{iBDD} at a target post-\ac{FEC} \ac{BER} of \(10^{-8}\) for PCs with different \([n,t]\) eBCH component codes, where \(n\) is the component-code length and \(t\) is the error-correcting capability.}
    \label{tab:compare}
    \centering
    \begin{tabular}{@{}l|ccccc@{}}
        \toprule
        \diagbox{Decoder}{\([n,t]\)}
        & \([128,2]\) & \([256,2]\) & \([512,2]\) & \([256,3]\) & \([512,3]\) \\
        \midrule
        \acs{RDRSD} ($\mathcal{J}=3$)
        & \(1.30\) & \(0.95\) & \(0.85\) & \(0.96\) & \(0.79\) \\
        \ac{DRSD}~\cite{miao2022JLT}
        & \(1.14\) & \(0.89\) & \(0.69\) & \(0.62\) & \(0.51\) \\
        \bottomrule
    \end{tabular}

    \vspace{2pt}
    {\footnotesize All gains are given in \si{dB}.}
\end{table}

\section{Error Floor Analysis}\label{sec:floor}

\subsection{Stall Patterns}
In this section, we define error-only stall patterns, which are the dominant mechanism behind error floors in iterative decoding of GPCs. An error floor is the high-\ac{SNR} regime in which the steep waterfall transition turns into a slow decrease of the post-decoding error rate. The shape of an error-only stall pattern depends on the GPC structure. For clarity, we use PCs to present the theoretical and simulation results.
\begin{definition}
An error-only stall pattern $\mathcal{S}$ is defined as a set of erroneous bit positions such that each involved \ac{CN} contains at least $t+1$ erroneous bits.
\end{definition}

For an error-only stall pattern in a PC, we define \(K\) and \(L\) as the numbers of involved rows and columns, respectively. Specifically, an involved row or column contains at least one error in the stall pattern. By ignoring positions outside the stall pattern, we can represent the stall pattern on a \(K\times L\) grid. For PCs, it is easy to see that the smallest error-only stall pattern involving \(K\) rows and \(L\) columns has size
\begin{equation}\label{eq:smallestS}
    |\mathcal{S}|_{\min} = \max\{K,L\}\cdot (t+1).
\end{equation}

A minimum error-only stall pattern is defined as the case $K=L=t+1$. For example, Fig.~\ref{fig:minimumSP} depicts a minimum stall pattern involving three row CNs and three column CNs.

\begin{figure}
    \centering
        \begin{tikzpicture}[scale=0.3]

        \tikzset{
  cross/.pic={
    \draw[line width=0.9pt, black] (-0.08,-0.08) -- (0.08,0.08);
    \draw[line width=0.9pt, black] (-0.08,0.08) -- (0.08,-0.08);
  }
}

  \draw[step=1cm,thick, black!90] (0,0) grid (7,6);

  \foreach \c in {0,2,6} {
    \foreach \r in {1,4,5} {
       \pic at (\c+0.5,\r+0.5) {cross};
    }
  }

  \draw[step=1cm,thick, black!90] (18,1) grid (21,4);

  \foreach \r in {1,2,3} {
    \foreach \c in {18,19,20} {
       \pic at (\c+0.5,\r+0.5) {cross};
    }
  }
  \node[draw=none] at (12,3.5) {$\rightarrow$};
\draw[decorate, decoration={brace, amplitude=2pt}]
  (18,4.2) -- (21,4.2)
  node[midway, yshift=10pt] {$L=3$};

\draw[decorate, decoration={brace, amplitude=2pt}]
   (21.2,4)--(21.2,1)
  node[midway, xshift=12pt, rotate=90] {$K=3$};

   \node[draw=none] at (3.5,-1) {\footnotesize (a) Minimum error-only};
   \node[draw=none] at (3.5,-2.0) {\footnotesize stall pattern};
    \node[draw=none] at (20,-1) {\footnotesize (b) Same pattern};
    \node[draw=none] at (20,-2.0) {\footnotesize with involved CNs};
\end{tikzpicture}

    \caption{A minimum error-only stall pattern for a PC based on $t=2$ component codes.}
    \label{fig:minimumSP}
\end{figure}
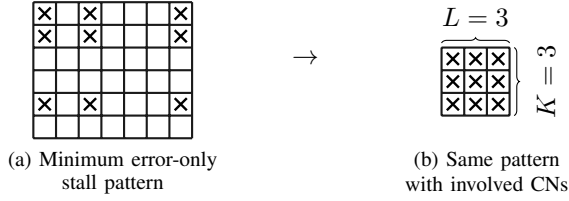

\subsection{Union Bound of Miscorrection-Free Decoding}

The error floor can be approximated by enumerating possible stall patterns and summing their contributions to the post-decoding \ac{BER}, assuming that a sufficient number of iterations has been performed to remove all errors outside the stall patterns. In this section, we first review the union-bound method for \ac{iBDD} proposed in~\cite{staircaseCode} and then extend this approach to miscorrection-free \ac{EaE} decoding.

\subsubsection{For iBDD}

For fixed \(K\) and \(L\), the \ac{BER} contribution of stall patterns involving \(K\) rows and \(L\) columns is approximated by
\begin{equation}\label{eq:iBDDfloor}
    \sum_{l=(t+1)\cdot \max \{K,L\}}^{KL}\binom{n}{K}\binom{n}{L}\frac{l}{n^2}M_{K,L}^l \cdot \delta^l,
\end{equation}
where \(\delta\) denotes the crossover probability of the \ac{BSC}, and $M_{K,L}^l$ is the number of stall patterns of size $l$ that can be arranged within $K$ fixed rows and $L$ fixed columns. To estimate the overall error floor, \eqref{eq:iBDDfloor} is evaluated over a range of values of $K$ and $L$, and the resulting contributions are accumulated. In practice, a typical choice is $t+1\leq K,L \leq 2t+3$, since larger values of $K$ and $L$ usually lead to negligible contributions due to the factor $\delta^l$.

In~\cite{staircaseCode}, an upper bound on \(M_{K,L}^l\) was proposed, while~\cite{holzbaur2019improved} provides both an exact computation method and an approximation.
Since the exact computation of \(M_{K,L}^l\) is computationally intensive, we use the sampling-based approximation in~\cite[eq.~(13)]{holzbaur2019improved}, which was shown to be close to the exact values.
The resulting union bound for the \ac{PC} with \([256,239,6]\) eBCH component code is shown by the \(\te=0\) union-bound curve in Fig.~\ref{fig:floor255codeUB}.

\subsubsection{For Ideal \acs{iEaED}}
For \ac{EaE} decoding described in Sec.~\ref{sec:EaED} followed by perfect miscorrection detection, we consider three types of stall patterns.

\begin{figure}
    \centering
    \input{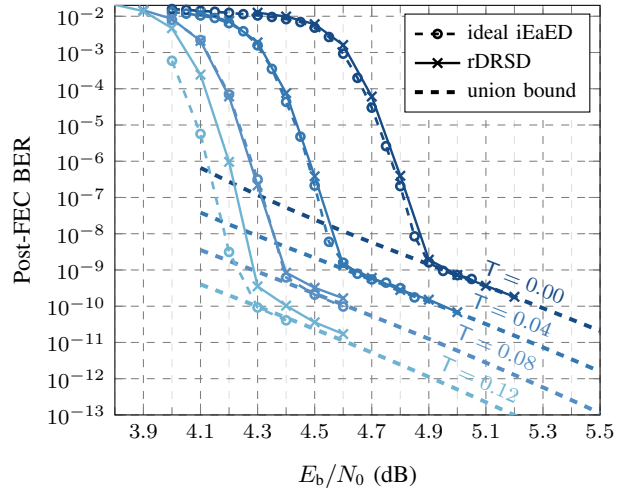}
    \caption{Comparison of simulated \acs{RDRSD}, ideal \acs{iEaED}, and the computed union bound for different erasure thresholds $\te$ for the product code with \([256,239,6]\) eBCH component code.}
    \label{fig:floor255codeUB}
\end{figure}

Type-I stall patterns contain errors whose positions already form an error-only stall pattern and may also contain erasures.
That is, the errors are located within \(K\) involved rows and \(L\) involved columns, and each involved row and column contains at least \(t+1\) errors, irrespective of the erasures inside the \(K\times L\) grid.
Such stall patterns are therefore dominated by errors.
An example is shown in Fig.~\ref{fig:EaEstallPatternReverse}-(a).

We next consider the effect of erasures within a Type-I stall pattern.
For fixed \(K\) and \(L\), the number of erasures is upper bounded by \(KL-\max\{K,L\}(t+1)\) according to~\eqref{eq:smallestS}.
Since the error floor is dominated by stall patterns with small \(K\) and \(L\), and each erasure contributes on average only half an error, the additional contribution due to erasures is small.
Consequently, the resulting error rate remains of the same order as that obtained when erasures inside the stall pattern are neglected.

Type-II stall patterns contain erasures and errors whose error positions alone do not form a stall pattern.
However, for each involved row and column, if \(U\) and \(E\) denote the number of errors and erasures, respectively, then
\(
    2U+E\geq \dmin .
\)
An example is shown in Fig.~\ref{fig:EaEstallPatternReverse}-(b) for $t=2$. For algebraic \ac{EaE} decoding, such patterns will give a non-negligible contribution to the error floor that is difficult to estimate. However, for the \ac{EaED} considered in this work, such patterns are less critical. As shown in~\cite{miao2022JLT}, a miscorrection-free ideal \ac{EaED} step resolves two errors and two erasures with probability $p=\frac{1}{2}$. Moreover, for values of \(U\) and \(E\) such that \(2U+E\) exceeds $\dmin$ only slightly, e.g., by $1$, $2$, or $3$, the probability $p$ remains close to $\frac{1}{2}$.
Since a stall pattern can be recovered with high probability once any of its involved CNs is recovered, repeated independent random trials across iterations make the survival probability \((1-p)^m\) decrease rapidly with the iteration number \(m\).

Type-III stall patterns consist only of erasures, e.g., as shown in Fig.~\ref{fig:EaEstallPatternReverse}-(c).
For pure erasures, the smallest stall pattern has size \(\dmin\times\dmin\), since each involved row and column must contain at least \(\dmin\) erasures.
Assuming independent erasures with probability \(\epsilonC=0.02\), which corresponds to PC decoding in the error-floor region, the contribution of the smallest Type-III stall patterns for the \([256,239,6]\)-eBCH code-based \ac{PC} is approximated by
\[
    \binom{256}{6}^2 \epsilonC^{36}
    \frac{36}{256^2}
    \approx 5.1\cdot 10^{-42}
\]
erasures per bit.
This contribution is negligible compared with the error-dominated stall-pattern contribution.

Under this miscorrection-free, error-dominated approximation, the error floor contribution for stall patterns on a $K\times L$ grid for ideal \acs{iEaED} is approximated as
\begin{equation}\label{eq:EaEfloor}
    \sum_{l=(t+1)\cdot \max \{K,L\}}^{KL}\binom{n}{K}\binom{n}{L}\frac{l}{n^2}M_{K,L}^l \cdot \deltaC^l.
\end{equation}
Compared with the ideal iBDD case in~\eqref{eq:iBDDfloor}, the only difference in the approximation is to replace the crossover probability $\delta$ of the BSC by the error probability $\deltaC$ of the EaE channel computed by~\eqref{eq:EaEprob}.

\begin{figure}
    \centering
    \begin{tikzpicture}[scale=0.4]
          \tikzset{
              cross/.pic={
                  \draw[line width=0.9pt, black] (-0.08,-0.08) -- (0.08,0.08);
                  \draw[line width=0.9pt, black] (-0.08,0.08) -- (0.08,-0.08);
              }
          }

          \tikzset{
              qmark/.pic={
                  \node[
                      inner sep=0pt,
                      font=\sffamily\bfseries\small,
                      text=black,
                  ] at (0,0) {?};
              }
          }

          \draw[step=1cm,thick, black!90] (0,0) grid (4,3);
          \foreach \x/\y in {0/0, 0/1, 0/2, 1/1,1/2,1/3, 2/0,2/2,2/3} {
              \pic at (\y+0.5,\x+0.5) {cross};
          }
          \foreach \x/\y in {0/3, 1/0, 2/1} {
              \pic at (\y+0.5,\x+0.5) {qmark};
          }

          \pgfmathsetmacro{\secondx}{6}

          \draw[step=1cm,thick, black!90] (\secondx,0) grid (\secondx+5,5);
          \foreach \x/\y in {0/0, 0/1, 1/1,1/2,2/2,2/3,3/3,3/4,4/4,4/0} {
              \pic at (\y+0.5+\secondx,\x+0.5) {cross};
          }

          \draw[step=1cm,thick, black!90] (\secondx,0) grid (\secondx+5,5);
          \foreach \x/\y in {0/2, 0/4, 1/0,1/3,2/1,2/4,3/2,3/0,4/3,4/1} {
              \pic at  (\y+0.5+\secondx,\x+0.5) {qmark};
          }

          \pgfmathsetmacro{\thirdx}{13}

          \draw[step=1cm,thick, black!90] (\thirdx,0) grid (\thirdx+6,6);
          \foreach \x in {0,1,...,5} {
              \foreach \y in {0,1,...,5} {
                  \pic at  (\y+0.5+\thirdx,\x+0.5) {qmark};
              }
          }

          \node[draw=none] at (2,-1) {\footnotesize (a) Type I};
          \node[draw=none] at (\secondx+3,-1) {\footnotesize (b) Type II};
          \node[draw=none] at (\thirdx+3,-1) {\footnotesize (c) Type III};

    \end{tikzpicture}
    \caption{Examples of Type-I, Type-II, and Type-III stall patterns for \ac{EaE} decoding of a PC with $t=2$. Cross marks indicate errors, and question marks indicate erasures.}
    \label{fig:EaEstallPatternReverse}
\end{figure}
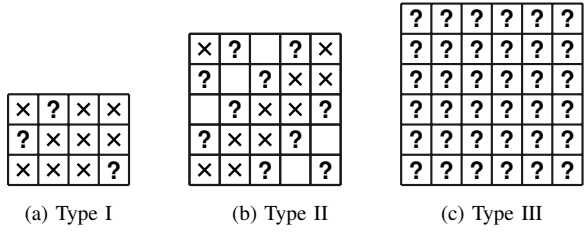

Fig.~\ref{fig:floor255codeUB} shows the union bound computed for different erasure thresholds \(\te\).
First, increasing \(\te\) lowers the error floor because the error probability
\(
    \deltaC = Q\left(\frac{\te+1}{\sigma}\right)
\)
of the EaE channel decreases with \(\te\), while the contribution of pure-erasure stall patterns remains negligible in the considered range of $\Eb/\No$.
Second, the performance of ideal \acs{iEaED} converges to the computed error-floor bound at high $\Eb/\No$.
Finally, the \acs{RDRSD} follows the behavior of ideal \acs{iEaED} closely, particularly for small \(\te\), where miscorrection detection is almost perfect.

\subsection{Error Floor with Miscorrections}
In~\cite{staircaseCode}, miscorrections are accounted for by replacing \(\delta\) in~\eqref{eq:iBDDfloor} with \(\delta+\zeta\), where \(\zeta\) denotes the per-bit miscorrection probability estimated by simulation under the assumption that miscorrections occur independently as Bernoulli events for each bit.
For the \acs{RDRSD}, however, this correction is not sufficiently accurate.
The miscorrection probability of the \acs{RDRSD} is typically much smaller than the input error and erasure probabilities.
For example, for the \([256,239,6]\) eBCH-code-based \ac{PC} at \(\Eb/\No=\SI{4.35}{dB}\), our numerical simulations estimate that the average number of miscorrected bits accumulated over \(20\) decoding iterations, normalized per code bit, is \(1.3\cdot 10^{-4}\). This is only in the order of \(1\%\) of the input error and erasure probabilities, \(\deltaC\) and \(\epsilonC\), respectively.
Such a small offset has almost no effect on the union bound computed from~\eqref{eq:EaEfloor}.
Nevertheless, as shown for the \(\te=0.12\) case in Fig.~\ref{fig:floor255codeUB}, the \acs{RDRSD} exhibits a noticeably higher error floor than the computed bound.
This suggests that the remaining miscorrections cannot be modeled as independent random errors.
To the best of our knowledge, no accurate analytical method is available for this case; hence, Monte-Carlo simulations are required to estimate the true error floor, while the miscorrection-free bound serves as a useful reference for guiding the simulation design.

\section{Stall Pattern Removal via Post-Processing}\label{sec:spr}

Many small error-only stall patterns can be identified from the remaining nonzero CN syndromes after iterative decoding. For example, with a $t=2$ component code, a minimum error-only stall pattern can be detected from the intersection of three failed row CNs and three failed column CNs. Conventional post-processing methods for such patterns can be broadly classified into flip-and-iterate and erase-and-iterate approaches. In both cases, rows and columns with nonzero syndromes are first marked. Then, selected bits at their intersections are flipped or erased, followed by additional \ac{iBDD} or iterative algebraic \ac{EaE} decoding iterations.

Flip-and-iterate post-processing has been applied to braided codes in~\cite{Jian2013} and to OFEC codes in~\cite{rapp2021error,lagendijk2025lowering}, while erase-and-iterate post-processing has been studied for \acp{HPC} in~\cite{emmadi2015half} and product codes in~\cite{condo2016stall}. In this section, we examine the success and failure conditions of both techniques and then propose a soft-aided post-processing method that resolves certain failure cases not addressed by existing approaches.

We evaluate the \([256,239,6]\) eBCH-code-based PC because this component code is widely used, e.g., in OFEC~\cite{openroadmOpenROADMMSA2021}, and its comparatively high PC error floor permits feasible Monte Carlo simulation.

\subsection{Correction Guarantee of Conventional Post-Processing Approaches}

Theorem~\ref{theorem:smallKL} unifies guarantees previously proven for separate cases~\cite{holzbaur2019improved,Jian2013,hager2018approaching} and characterizes the stall patterns corrected by both approaches.

\begin{theorem}\label{theorem:smallKL}
Consider a product code with eBCH component codes of minimum distance
\(\dmin=2t+2\). Assume that no errors exist outside the error-only stall pattern, that
the error-only stall pattern is fully identified by \(K\) failed row CNs and \(L\) failed
column CNs, and that no undetectable component-code errors are present.
If \(L\leq 2t+1\), then both flip-and-iterate and erase-and-iterate
post-processing resolve the stall pattern when the subsequent decoding
iteration starts with row decoding. Similarly, if \(K\leq 2t+1\), both methods
resolve the stall pattern when the subsequent decoding iteration starts with
column decoding.
\end{theorem}
\begin{proof}
    Let \(r_i\) denote the number of errors in the \(i\)-th involved row and
    \(s_j\) the number of errors in the \(j\)-th involved column.
    From the definition of error-only stall patterns,
    \(r_i\geq t+1\) and \(s_j\geq t+1\).

    First consider the flip-and-iterate method.
    If \(L\leq 2t+1\), then after flipping all \(K L\) intersection bits, the
    number of errors in the \(i\)-th involved row becomes \(L-r_i\).
    Since \(r_i\geq t+1\), we have
    \[
        L-r_i \leq L-(t+1) \leq t .
    \]
    Hence, all involved rows are correctable by one row-decoding half-iteration.
    The case for \(K\leq 2t+1\) follows analogously.

    Now consider the erase-and-iterate method.
    If \(L\leq 2t+1=\dmin-1\), then after erasing all intersection bits, each
    involved row contains at most \(L<\dmin\) erasures and no residual errors in
    the marked positions.
    Therefore, each involved row is uniquely recoverable by algebraic erasure
    decoding.
    Analogously, if \(K\leq 2t+1=\dmin-1\), each involved column contains fewer
    than \(\dmin\) erasures and is recoverable.
\end{proof}

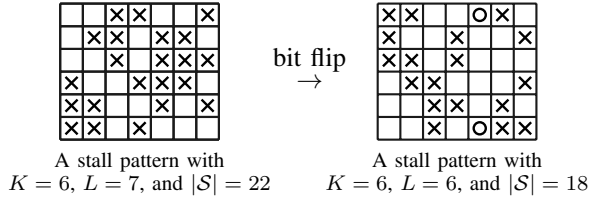
\begin{figure}
    \centering
    \begin{tikzpicture}[scale=0.3]
       \tikzset{
           cross/.pic={
               \draw[line width=0.9pt, black] (-0.08,-0.08) -- (0.08,0.08);
               \draw[line width=0.9pt, black] (-0.08,0.08) -- (0.08,-0.08);
           },
           crossb/.pic={
               \draw[line width=0.9pt, black] (0,0) circle[radius=0.08];
           }
       }

       \draw[step=1cm,thick, black!90] (0,0) grid (7,6);
       \foreach \x/\y in {0/0, 0/1, 0/3,1/0,1/1,1/4,1/6,2/0,2/3,2/4,2/5,3/2,3/4,3/5,3/6,4/1,4/2,4/4,4/5,5/2,5/3,5/6} {
           \pic at (\y+0.5,\x+0.5) {cross};
       }
       \node[draw=none] at (3.5,-1) {\footnotesize A stall pattern with };
       \node[draw=none] at (3.5,-2) {\footnotesize $K=6$, $L=7$, and $|\mathcal{S}|=22$};

       \node[draw=none] at (11,2.5) {$\rightarrow$};
       \node[draw=none] at (11,3.5) {bit flip};

       \draw[step=1cm,thick, black!90] (14,0) grid (21,6);

       \draw[step=1cm,thick, black!90] (0,0) grid (7,6);
       \foreach \x/\y in {0/2, 0/5,0/6,1/2,1/3,1/5,2/1,2/2,2/6,3/0,3/1,3/3,4/0,4/3,4/6,5/0,5/1,5/5} {
           \pic at (\y+14.5,\x+0.5) {cross};
       }

       \foreach \x/\y in {0/4,5/4} {
           \pic at (\y+14.5,\x+0.5) {crossb};
       }
       \node[draw=none] at (17.5,-1) {\footnotesize A stall pattern with };
       \node[draw=none] at (17.5,-2) {\footnotesize $K=6$, $L=6$, and $|\mathcal{S}|=18$};

    \end{tikzpicture}
    \caption{A large stall pattern whose complementary pattern is again a stall pattern of smaller size. A cross represents an error within a stall pattern, while a circle represents an error outside a stall pattern.}
    \label{fig:stallPatternReverse67}
\end{figure}

A direct consequence of Theorem~\ref{theorem:smallKL} is that, for PCs with
eBCH component codes, all error-only stall patterns satisfying
\(
    {|\mathcal{S}| < (t+1)\dmin}
\)
can be resolved by both flip-and-iterate and erase-and-iterate post-processing.
Since each involved row and column contains at least \(t+1\) errors, we have
\[
    |\mathcal{S}| \geq (t+1)\max\{K,L\}.
\]
Together with \( |\mathcal{S}| < (t+1)\dmin \), this implies
\(
    {\max\{K,L\} < \dmin }.
\)
Hence, no undetectable component-code error can be present as an undetectable
error would require at least \(\dmin\) erroneous positions. Since \(\max\{K,L\}<\dmin=2t+2\), we have
\(\max\{K,L\}\leq 2t+1\), and the stall pattern is correctable by
Theorem~\ref{theorem:smallKL}.

When \(\min\{K,L\}\geq 2t+2\), such a guarantee no longer exists.

\newcommand{\Mr}{\mathcal{M}_{\text{r}}}
\newcommand{\Mc}{\mathcal{M}_{\text{c}}}
\newcommand{\Fri}{\mathcal{F}_{\text{r},i}}
\newcommand{\Fci}{\mathcal{F}_{\text{c},i}}

\begin{algorithm}[t]
\footnotesize
\DontPrintSemicolon
\caption{\acs{RDRSD} with Post-Processing for Stall Pattern Removal}
\label{algo:PP}

\KwIn{$\bm{y}\in\ZQO^N$}
\KwOut{$\bm{y}$}

Run regular \acs{RDRSD} iterations and record the set of flipped bits for each CN in
$\mathcal{F}_{\text{r},1},\ldots,\mathcal{F}_{\text{r},n}$ and
$\mathcal{F}_{\text{c},1},\ldots,\mathcal{F}_{\text{c},n}$\;

Run one iBDD iteration and record all CNs that either fail decoding or correct bits
into $\Mr$ and $\Mc$\;

\If{$\Mr=\emptyset$ \textbf{and} $\Mc=\emptyset$}{
    \Return{$\bm{y}$}\;
}

\If{$\Mr=\emptyset$}{
    \For{$i=1,\ldots,n$}{
        \If{$\Fri\subseteq \Mc$}{
            $\Mr \gets \Mr \cup \{i\}$\;
        }
    }
}

\If{$\Mc=\emptyset$}{
    \For{$i=1,\ldots,n$}{
        \If{$\Fci\neq\emptyset$ \textbf{and} $\Fci\subseteq \Mr$}{
            $\Mc \gets \Mc \cup \{i\}$\;
        }
    }
}

\If{
    $\bigl(|\Mr|>\dmin+2 \ \textbf{and}\ |\Mc|\geq\dmin\bigr)$
    \textbf{or}
    $\bigl(|\Mr|\geq\dmin \ \textbf{and}\ |\Mc|>\dmin+2\bigr)$
}{
    \Return{$\bm{y}$}\;
}

\eIf{$|\Mr|<\dmin$ \textbf{or} $|\Mc|<\dmin$}{
    Flip all bits associated with CN pairs in $\Mr\times\Mc$\;
}{
    Draw \(R_i\) uniformly from \(\{0,1,\ldots,2^q-1\}\) for each bit \(i\)
    associated with a CN pair in \(\Mr\times\Mc\)\;
    Erase bit \(i\) if \(R_i\geq 2D_i\)\;
}

Run one \acs{iEaED} iteration, correcting only errors or erasures associated with a CN in $\Mr$ or $\Mc$\;
Run one regular \acs{iEaED} iteration\;

$\Mr\gets\emptyset$, $\Mc\gets\emptyset$\;
Record all row and column CNs with nonzero syndrome into $\Mr$ and $\Mc$, respectively\;

Repeat Lines~13--21 once more as specified above\;

\Return{$\bm{y}$}\;

\end{algorithm}

\begin{figure*}[t]
\centering
\begin{tabular}{@{}c@{\quad}c@{}@{\quad}c@{}@{\quad}c@{}}
        \multicolumn{4}{c}{\begin{tikzpicture}
\begin{axis}[
    hide axis,
    scale only axis,
    width=0pt,
    height=0pt,
    xmin=0, xmax=1,
    ymin=0, ymax=1,
    legend columns=4,
    legend style={
        draw=black,
        at={(0.5,0)},
        anchor=center,
        /tikz/every even column/.append style={column sep=0.4cm},
        font = \footnotesize
    },
]

    \addlegendimage{color=black, line width=1pt, mark=o, mark options={solid,fill=white, mark size=1.5pt}}
    \addlegendentry{rDRSD}

        \addlegendimage{color=KITpalegreen, line width=1pt, mark=o, mark options={ solid, mark size=1pt}}
    \addlegendentry{rDRSD + conventional PP~\cite{holzbaur2019improved}}

    \addlegendimage{color=magenta, line width=1pt, mark=x, mark options={solid,fill=white, mark size=2.5pt}}
    \addlegendentry{rDRSD + proposed PP}

\end{axis}
\end{tikzpicture}} \\[-0.5ex]
		\begin{tikzpicture}
		\pgfplotsset{grid style={dashed, gray}}
		\pgfplotsset{every tick label/.append style={font=\footnotesize}}

		\begin{axis}[%
  xmin=4.1,
  xmax=4.3,
  ymode=log,
  ymin=1e-7,
  ymax=0.1,
  yminorticks=true,
  axis background/.style={fill=white, mark size=1.5pt},
  xmajorgrids,
  xminorgrids,
  ymajorgrids,
  yminorgrids,
  width=4cm,
  height=4cm,
  xtick={3.5,3.7,...,10.5},
  minor x tick num=1,
  minor grid style={gray!25},%
  ytick={1,0.1,0.01,0.001,1e-4,1e-5,1e-6,1e-7,1e-8,1e-9,1e-10,1e-11,1e-12,1e-13,1e-14,1e-15},
  xlabel={$\Eb/\No$ (dB)},
  ylabel={FER},
  label style={font=\small},
  legend cell align={left},
  legend style={{anchor = south west}, draw=none, fill opacity=0.7, text opacity = 1,legend columns=1,font=\footnotesize, row sep = 0pt}
		]

\addplot [color=black, line width=1pt, mark=o, mark options={solid,fill=white, mark size=1.5pt}]table[x=EbNo,y=FER, col sep=semicolon,row sep=crcr] {
  EbNo;        EsNo;       delta; ErasureProb;  totalFrame;    FE;         FER;         BER;      SPsize;          BE;  throughput;    mis rate\\
   4.1;     3.50316;  0.00837965;   0.0243882;       25600;  2008;   0.0784375; 0.000234492;     195.923;      393412; 1.05514e+09; 0.000623453\\ %
   4.2;     3.60316;  0.00776787;   0.0234695;       25600;    57;  0.00222656; 9.06885e-07;      26.693;      1521.5; 1.15613e+09; 0.000248367\\ %
   4.3;     3.70316;  0.00718933;   0.0225605;       51200;    80;   0.0015625; 2.15471e-07;      9.0375;         723; 1.18265e+09; 0.000163651\\ %
   4.4;     3.80316;   0.0066431;   0.0216623;       51200;    65;  0.00126953; 1.74344e-07;           9;         585; 1.13499e+09; 0.000117783\\ %
   4.5;     3.90316;   0.0061282;   0.0207759;       76800;    68; 0.000885417; 1.21593e-07;           9;         612; 1.12618e+09;  8.9432e-05\\ %
   4.6;     4.00316;  0.00564362;   0.0199023;       51200;    54;  0.00105469; 1.44839e-07;           9;         486; 1.11011e+09; 7.10613e-05\\ %
   4.7;     4.10316;  0.00518836;   0.0190426;       76800;    55; 0.000716146; 9.83477e-08;           9;         495; 1.07128e+09; 5.86891e-05\\ %
};%

\addplot [color=KITpalegreen, dashed, line width=1pt, mark=o, mark options={solid,fill=white, mark size=1.5pt}]table[x=EbNo, y=FER, col sep=semicolon,row sep=crcr] {
  EbNo;        EsNo;       delta; ErasureProb;  totalFrame;    FE;         FER;         BER;      SPsize;          BE;  throughput;    mis rate\\
   4.1;     3.50316;  0.00837965;   0.0243882;       25600;  1738;   0.0678906; 0.000238346;     230.079;      399878; 1.13798e+09; 0.000636863\\ %
   4.2;     3.60316;  0.00776787;   0.0234695;      128000;    57; 0.000445313; 8.18372e-07;     120.439;        6865; 1.13871e+09; 0.000246426\\ %
   4.3;     3.70316;  0.00718933;   0.0225605; 1.27053e+08;    50; 3.93537e-07; 1.64895e-10;       27.46;        1373;    1.09e+09; 0.000163608\\ %
};%

\addplot [color=magenta, line width=1pt, mark=x, mark options={solid,fill=white, mark size=2.5pt}]table[x=EbNo, y=FER, col sep=semicolon,row sep=crcr] {
  EbNo;        EsNo;       delta; ErasureProb;  totalFrame;    FE;         FER;         BER;      SPsize;          BE;  throughput;    mis rate\\
   4.1;     3.50316;  0.00837965;   0.0243882;       25600;  1704;   0.0665625;  0.00023631;     232.666;      396463; 1.00645e+09; 0.000631718\\ %
   4.2;     3.60316;  0.00776787;   0.0234695;      179200;    53; 0.000295759; 6.32916e-07;     140.245;        7433; 1.05352e+09; 0.000246319\\ %
   4.3;     3.70316;  0.00718933;   0.0225605; 1.81171e+08;    39; 2.15266e-07; 2.00872e-10;     61.1538;        2385; 1.09925e+09; 0.000163594\\ %
};%

\node[anchor=south west, font=\footnotesize] at (rel axis cs:0,0) {(a) $(3,3)$};
\end{axis}
\end{tikzpicture}&		\begin{tikzpicture}
		\pgfplotsset{grid style={dashed, gray}}
		\pgfplotsset{every tick label/.append style={font=\footnotesize}}

		\begin{axis}[%
  xmin=4.1,
  xmax=4.3,
  ymode=log,
  ymin=1e-6,
  ymax=0.1,
  yminorticks=true,
  axis background/.style={fill=white, mark size=1.5pt},
  xmajorgrids,
  xminorgrids,
  ymajorgrids,
  yminorgrids,
  width=4cm,
  height=4cm,
  xtick={3.5,3.7,...,10.5},
  minor x tick num=1,
  minor grid style={gray!25},%
  ytick={1,0.1,0.01,0.001,1e-4,1e-5,1e-6,1e-7,1e-8,1e-9,1e-10,1e-11,1e-12,1e-13,1e-14,1e-15},
  xlabel={$\Eb/\No$ (dB)},
  label style={font=\small},
  legend cell align={left},
  legend style={{anchor = south west}, draw=none, fill opacity=0.7, text opacity = 1,legend columns=1,font=\footnotesize, row sep = 0pt}
		]

\addplot [color=black, line width=1pt, mark=o, mark options={solid,fill=white, mark size=1.5pt}]table[x=EbNo,y=FER, col sep=semicolon,row sep=crcr] {
  EbNo;        EsNo;       delta; ErasureProb;  totalFrame;    FE;         FER;         BER;      SPsize;          BE;  throughput;    mis rate\\
   4.1;     3.50316;  0.00837965;   0.0243882;       25600;  2377;   0.0928516; 0.000271171;     191.397;      454950; 1.02715e+09;   0.0006548\\ %
   4.2;     3.60316;  0.00776787;   0.0234695;       25600;   189;  0.00738281; 2.17915e-06;     19.3439;        3656;   1.031e+09; 0.000250806\\ %
   4.3;     3.70316;  0.00718933;   0.0225605;       25600;   131;  0.00511719; 7.47442e-07;     9.57252;        1254; 1.09479e+09; 0.000164653\\ %
   4.4;     3.80316;   0.0066431;   0.0216623;       25600;   119;  0.00464844; 6.45518e-07;     9.10084;        1083;  9.6211e+08; 0.000118795\\ %
   4.5;     3.90316;   0.0061282;   0.0207759;       25600;    97;  0.00378906; 5.32866e-07;     9.21649;         894; 1.08664e+09; 9.00197e-05\\ %
   4.6;     4.00316;  0.00564362;   0.0199023;       25600;    88;   0.0034375; 4.84586e-07;     9.23864;         813; 1.08931e+09; 7.11077e-05\\ %
   4.7;     4.10316;  0.00518836;   0.0190426;       25600;    71;  0.00277344; 3.88026e-07;     9.16901;         651;  1.0722e+09; 5.82463e-05\\ %
};%

\addplot [color=KITpalegreen, dashed, line width=1pt, mark=o, mark options={solid,fill=white, mark size=1.5pt}]table[x=EbNo, y=FER, col sep=semicolon,row sep=crcr] {
  EbNo;        EsNo;       delta; ErasureProb;  totalFrame;    FE;         FER;         BER;      SPsize;          BE;  throughput;    mis rate\\
   4.1;     3.50316;  0.00837965;   0.0243882;       25600;  1919;   0.0749609; 0.000259028;      226.46;      434577; 1.08351e+09; 0.000660933\\ %
   4.2;     3.60316;  0.00776787;   0.0234695;      102400;    53; 0.000517578; 9.21637e-07;     116.698;        6185; 1.02646e+09; 0.000248965\\ %
   4.3;     3.70316;  0.00718933;   0.0225605;   2.176e+06;    50; 2.29779e-05; 6.81597e-09;       19.44;         972; 9.87702e+08; 0.000164746\\ %
};%

\addplot [color=magenta, line width=1pt, mark=x, mark options={solid,fill=white, mark size=2.5pt}]table[x=EbNo, y=FER, col sep=semicolon,row sep=crcr] {
  EbNo;        EsNo;       delta; ErasureProb;  totalFrame;    FE;         FER;         BER;      SPsize;          BE;  throughput;    mis rate\\
   4.1;     3.50316;  0.00837965;   0.0243882;       25600;  1862;   0.0727344;  0.00025613;     230.781;      429715; 1.06244e+09;  0.00065086\\ %
   4.2;     3.60316;  0.00776787;   0.0234695;      102400;    53; 0.000517578; 9.31472e-07;     117.943;        6251; 1.12238e+09; 0.000249698\\ %
   4.3;     3.70316;  0.00718933;   0.0225605;  1.2032e+07;    50; 4.15559e-06; 1.67654e-09;       26.44;        1322; 1.10091e+09; 0.000164679\\ %
};%

\node[anchor=south west, font=\footnotesize] at (rel axis cs:0,0) {(b) $(3,4)$};
\end{axis}
\end{tikzpicture}&		\begin{tikzpicture}
		\pgfplotsset{grid style={dashed, gray}}
		\pgfplotsset{every tick label/.append style={font=\footnotesize}}

		\begin{axis}[%
  xmin=4.1,
  xmax=4.3,
  ymode=log,
  ymin=1e-5,
  ymax=0.1,
  yminorticks=true,
  axis background/.style={fill=white, mark size=1.5pt},
  xmajorgrids,
  xminorgrids,
  ymajorgrids,
  yminorgrids,
  width=4cm,
  height=4cm,
  xtick={3.5,3.7,...,10.5},
  minor x tick num=1,
  minor grid style={gray!25},%
  ytick={1,0.1,0.01,0.001,1e-4,1e-5,1e-6,1e-7,1e-8,1e-9,1e-10,1e-11,1e-12,1e-13,1e-14,1e-15},
  xlabel={$\Eb/\No$ (dB)},
  label style={font=\small},
  legend cell align={left},
  legend style={{anchor = south west}, draw=none, fill opacity=0.7, text opacity = 1,legend columns=1,font=\footnotesize, row sep = 0pt}
		]

\addplot [color=black, line width=1pt, mark=o, mark options={solid,fill=white, mark size=1.5pt}]table[x=EbNo,y=FER, col sep=semicolon,row sep=crcr] {
  EbNo;        EsNo;       delta; ErasureProb;  totalFrame;    FE;         FER;         BER;      SPsize;          BE;  throughput;    mis rate\\
   4.1;     3.50316;  0.00837965;   0.0243882;       25600;  2824;    0.110312; 0.000300903;     178.765;      504832; 1.02419e+09; 0.000684866\\ %
   4.2;     3.60316;  0.00776787;   0.0234695;       25600;   371;   0.0144922; 3.19839e-06;     14.4636;        5366; 1.02117e+09; 0.000250934\\ %
   4.3;     3.70316;  0.00718933;   0.0225605;       25600;   279;   0.0108984; 1.57535e-06;     9.47312;        2643; 1.04777e+09; 0.000166318\\ %
   4.4;     3.80316;   0.0066431;   0.0216623;       25600;   267;   0.0104297;  1.5074e-06;     9.47191;        2529;  1.0608e+09; 0.000119793\\ %
   4.5;     3.90316;   0.0061282;   0.0207759;       25600;   180;  0.00703125; 1.02282e-06;     9.53333;        1716;  1.0379e+09; 9.06843e-05\\ %
   4.6;     4.00316;  0.00564362;   0.0199023;       25600;   157;  0.00613281; 8.81255e-07;      9.4172;      1478.5; 1.04112e+09; 7.06893e-05\\ %
   4.7;     4.10316;  0.00518836;   0.0190426;       25600;   153;  0.00597656; 8.38637e-07;     9.19608;        1407; 1.06838e+09; 5.84674e-05\\ %
};%

\addplot [color=KITpalegreen, dashed, line width=1pt, mark=o, mark options={solid,fill=white, mark size=1.5pt}]table[x=EbNo, y=FER, col sep=semicolon,row sep=crcr] {
  EbNo;        EsNo;       delta; ErasureProb;  totalFrame;    FE;         FER;         BER;      SPsize;          BE;  throughput;    mis rate\\
   4.1;     3.50316;  0.00837965;   0.0243882;       25600;  2192;    0.085625; 0.000290267;     222.166;      486988; 1.08443e+09; 0.000695071\\ %
   4.2;     3.60316;  0.00776787;   0.0234695;       76800;    67; 0.000872396; 1.08957e-06;     81.8507;        5484; 9.69919e+08; 0.000250796\\ %
   4.3;     3.70316;  0.00718933;   0.0225605;      384000;    52; 0.000135417; 3.71933e-08;          18;         936; 1.06647e+09; 0.000165762\\ %
};%

\addplot [color=magenta, line width=1pt, mark=x, mark options={solid,fill=white, mark size=2.5pt}]table[x=EbNo, y=FER, col sep=semicolon,row sep=crcr] {
  EbNo;        EsNo;       delta; ErasureProb;  totalFrame;    FE;         FER;         BER;      SPsize;          BE;  throughput;    mis rate\\
   4.1;     3.50316;  0.00837965;   0.0243882;       25600;  2104;   0.0821875; 0.000286046;     228.092;      479906; 1.04662e+09; 0.000680564\\ %
   4.2;     3.60316;  0.00776787;   0.0234695;      102400;    57; 0.000556641; 7.97361e-07;     93.8772;        5351; 1.11258e+09; 0.000251254\\ %
   4.3;     3.70316;  0.00718933;   0.0225605;   1.408e+06;    51; 3.62216e-05; 1.37307e-08;     24.8431;        1267; 1.08551e+09;  0.00016562\\ %
};%
\node[anchor=south west, font=\footnotesize] at (rel axis cs:0,0) {(c) $(3,5)$};
\end{axis}
\end{tikzpicture}&		\begin{tikzpicture}
		\pgfplotsset{grid style={dashed, gray}}
		\pgfplotsset{every tick label/.append style={font=\footnotesize}}

		\begin{axis}[%
  xmin=4.1,
  xmax=4.6,
  ymode=log,
  ymin=1e-4,
  ymax=0.1,
  yminorticks=true,
  axis background/.style={fill=white, mark size=1.5pt},
  xmajorgrids,
  xminorgrids,
  ymajorgrids,
  yminorgrids,
  width=4cm,
  height=4cm,
  xtick={3.5,3.7,...,10.5},
  minor x tick num=1,
  minor grid style={gray!25},%
  ytick={1,0.1,0.01,0.001,1e-4,1e-5,1e-6,1e-7,1e-8,1e-9,1e-10,1e-11,1e-12,1e-13,1e-14,1e-15},
  xlabel={$\Eb/\No$ (dB)},
  label style={font=\small},
  legend cell align={left},
legend style={
    at={(1.1,0.95)},
    anchor=north west,
    draw=black,
    fill opacity=1,
    text opacity=1,
    legend columns=1,
    font=\footnotesize,
    row sep=0pt
}
		]

\addplot [color=black, line width=1pt, mark=o, mark options={solid,fill=white, mark size=1.5pt}]table[x=EbNo,y=FER, col sep=semicolon,row sep=crcr] {
  EbNo;        EsNo;       delta; ErasureProb;  totalFrame;    FE;         FER;         BER;      SPsize;          BE;  throughput;    mis rate\\
   4.1;     3.50316;  0.00837965;   0.0243882;       25600;  3314;    0.129453;  0.00031224;     158.072;      523852; 1.02179e+09; 0.000691476\\ %
   4.2;     3.60316;  0.00776787;   0.0234695;       25600;   627;   0.0244922; 5.93603e-06;     15.8836;        9959; 1.05525e+09; 0.000252447\\ %
   4.3;     3.70316;  0.00718933;   0.0225605;       25600;   529;   0.0206641; 3.72469e-06;     11.8129;        6249; 1.15373e+09; 0.000167665\\ %
   4.4;     3.80316;   0.0066431;   0.0216623;       25600;   436;   0.0170313; 3.01719e-06;     11.6101;        5062; 1.17813e+09; 0.000118732\\ %
   4.5;     3.90316;   0.0061282;   0.0207759;       25600;   362;   0.0141406; 2.57105e-06;     11.9157;      4313.5; 1.17739e+09; 9.09936e-05\\ %
   4.6;     4.00316;  0.00564362;   0.0199023;       25600;   317;   0.0123828; 2.18987e-06;     11.5899;        3674; 1.18187e+09; 7.33459e-05\\ %
   4.7;     4.10316;  0.00518836;   0.0190426;       25600;   277;   0.0108203; 2.02179e-06;     12.2455;        3392; 1.10687e+09; 5.93913e-05\\ %
};%

\addplot [color=KITpalegreen, dashed, line width=1pt, mark=o, mark options={solid,fill=white, mark size=1.5pt}]table[x=EbNo, y=FER, col sep=semicolon,row sep=crcr] {
  EbNo;        EsNo;       delta; ErasureProb;  totalFrame;    FE;         FER;         BER;      SPsize;          BE;  throughput;    mis rate\\
   4.1;     3.50316;  0.00837965;   0.0243882;       25600;  2321;   0.0906641; 0.000289038;     208.929;      484925; 9.56306e+08; 0.000692338\\ %
   4.2;     3.60316;  0.00776787;   0.0234695;       25600;    64;      0.0025; 1.87993e-06;     49.2812;        3154; 1.04127e+09; 0.000252151\\ %
   4.3;     3.70316;  0.00718933;   0.0225605;       51200;    63;  0.00123047; 4.81009e-07;      25.619;        1614; 1.03346e+09;  0.00016584\\ %
   4.4;     3.80316;   0.0066431;   0.0216623;       51200;    57;  0.00111328; 4.77433e-07;     28.1053;        1602; 1.12423e+09; 0.000119774\\ %
   4.5;     3.90316;   0.0061282;   0.0207759;       76800;    50; 0.000651042; 2.52128e-07;       25.38;        1269; 1.12599e+09; 9.07499e-05\\ %
   4.6;     4.00316;  0.00564362;   0.0199023;       76800;    62; 0.000807292; 3.73125e-07;     30.2903;        1878; 1.10664e+09; 7.22724e-05\\ %
   4.7;     4.10316;  0.00518836;   0.0190426;      102400;    56; 0.000546875; 2.29329e-07;     27.4821;        1539; 1.13182e+09;   5.906e-05\\ %
};%

\addplot [color=magenta, line width=1pt, mark=x, mark options={solid,fill=white, mark size=2.5pt}]table[x=EbNo, y=FER, col sep=semicolon,row sep=crcr] {
  EbNo;        EsNo;       delta; ErasureProb;  totalFrame;    FE;         FER;         BER;      SPsize;          BE;  throughput;    mis rate\\
   4.1;     3.50316;  0.00837965;   0.0243882;       25600;  2267;   0.0885547; 0.000298957;     221.247;      501567; 1.10349e+09; 0.000693964\\ %
   4.2;     3.60316;  0.00776787;   0.0234695;       25600;    66;  0.00257813; 2.06411e-06;     52.4697;        3463; 1.03065e+09; 0.000249906\\ %
   4.3;     3.70316;  0.00718933;   0.0225605;       51200;    61;  0.00119141; 5.87404e-07;     32.3115;        1971; 1.06455e+09; 0.000166655\\ %
   4.4;     3.80316;   0.0066431;   0.0216623;       76800;    66; 0.000859375; 4.33326e-07;     33.0455;        2181;   1.162e+09; 0.000119661\\ %
   4.5;     3.90316;   0.0061282;   0.0207759;       76800;    63; 0.000820312; 3.93391e-07;     31.4286;        1980;  1.1211e+09; 9.08444e-05\\ %
   4.6;     4.00316;  0.00564362;   0.0199023;       76800;    50; 0.000651042; 3.08752e-07;       31.08;        1554; 1.15423e+09; 7.21606e-05\\ %
   4.7;     4.10316;  0.00518836;   0.0190426;       76800;    55; 0.000716146; 3.59416e-07;     32.8909;        1809; 1.06417e+09; 5.98524e-05\\ %
};%
\node[anchor=south west, font=\footnotesize] at (rel axis cs:0,0) {(d) $(4,4)$};
\end{axis}
\end{tikzpicture}\\
\end{tabular}
    \caption{Post-decoding FER of the \([256,239,6]\) eBCH-code-based \ac{PC} under rejection-sampling-based evaluation for injected stall patterns with dimensions \({(K,L)\in\{(3,3),(3,4),(3,5),(4,4)\}}\). The proposed soft-aided post-processing (PP) technique is compared with conventional post-processing~\cite{holzbaur2019improved} and \acs{RDRSD} without post-processing.}
    \label{fig:ISPP_all}
\end{figure*}
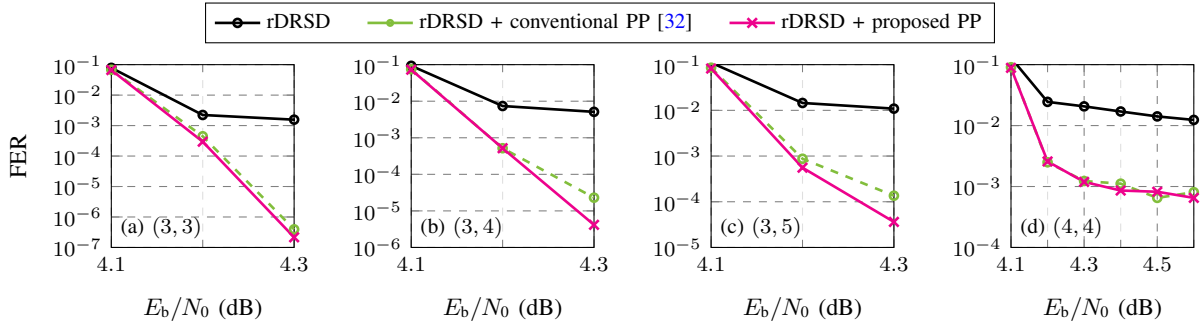

\subsection{Dealing with Large Stall Patterns}\label{subsec:largeStallPatterns}
Large stall patterns outside the scope of Theorem~\ref{theorem:smallKL} may lead to non-deterministic post-processing outcomes due to several failure mechanisms. We next analyze these mechanisms and discuss possible approaches to mitigate them.

First, if the number of preceding decoding iterations is insufficient, residual errors may remain outside the identified stall pattern. In this case, the marked CN set is inaccurate, leading to an unpredictable post-processing outcome. This issue is most pronounced when the \ac{SNR} is not sufficiently high. To mitigate it, we skip post-processing when the detected stall pattern is too large. Specifically, in our implementation, post-processing is skipped if \(K\geq \dmin\) and \(L\geq \dmin\), and at least one of \(K\) and \(L\) exceeds \(\dmin+2\).

Second, undetectable miscorrections can cause some erroneous CNs to remain unmarked.
This issue can be partially reduced by enlarging the marked CN set.
One approach, proposed in~\cite{hager2018approaching}, records the corrected positions of each CN during iterative decoding.
A row CN is then added to the marked set if its corrected positions form a subset of the currently marked column CNs; column CNs are treated analogously.
Another approach, proposed in~\cite{lagendijk2025lowering}, adds CNs that performed corrections in the last iteration to the marked set, thereby helping to identify CNs involved in miscorrection loops. We adopt both methods in our post-processing step. In Algorithm~\ref{algo:PP}, the subset-based enlargement is represented by the two conditional loops after the initial marking step, while CNs that fail or correct bits in the probing \ac{iBDD} iteration are recorded before these loops.

Third, in the presence of large stall patterns, subsequent \ac{iBDD}/\acs{iEaED} iterations after the bit flipping/erasing may cause miscorrections, which further introduce new errors outside of the current stall pattern.
This effect can be partly mitigated by restricting the following corrections to the marked CNs.

Finally, for large stall patterns, flipping all marked bits may transform the original pattern into another stall pattern.
An example is shown in Fig.~\ref{fig:stallPatternReverse67}.
Similarly, erase-and-iterate post-processing may create a Type-III stall pattern consisting only of erasures.
This issue can be mitigated by allowing multiple rounds of post-processing steps, as also used in~\cite{holzbaur2019improved,lagendijk2025lowering}.
Another approach is to only flip part of the marked bits, e.g., by flipping only one row of a large marked stall pattern before resuming \ac{iBDD}~\cite{holzbaur2019improved}.
To address this issue, we propose a soft-aided partial-erasure approach.
For each bit \(i\) located at the intersection of a marked row CN and a marked column CN, a random integer \(R_i\) is drawn uniformly from \(\{0,1,\ldots,2^q-1\}\).
The bit is erased if
\(
    R_i \geq 2D_i,
\)
where \(D_i\) denotes the bit \ac{DRS}.
Thus, bits with smaller \ac{DRS} values are more likely to be erased, while highly reliable bits are preserved with high probability.

Motivated by the preceding enumeration of failure mechanisms, Algorithm~\ref{algo:PP} summarizes the proposed post-processing procedure, which combines existing techniques with the modifications introduced in this work. The algorithm is applied after the regular \(L\) iterations of \acs{RDRSD}. During post-processing, miscorrection detection is disabled to avoid the effect of persistently misidentified anchor bits. Otherwise, such anchor bits may repeatedly reject the corrections needed to break a stall pattern.

\subsection{Evaluation of the Post-Processing Techniques}

\begin{figure}[t]
    \definecolor{mygreen}{rgb}{0.69, 0.87, 0.541}%
    \definecolor{myblue}{rgb}{0,0.4470,0.7410}
    \definecolor{myblack}{rgb}{0.2,0.2,0.2}
    \begin{tikzpicture}
        \centering
        \input{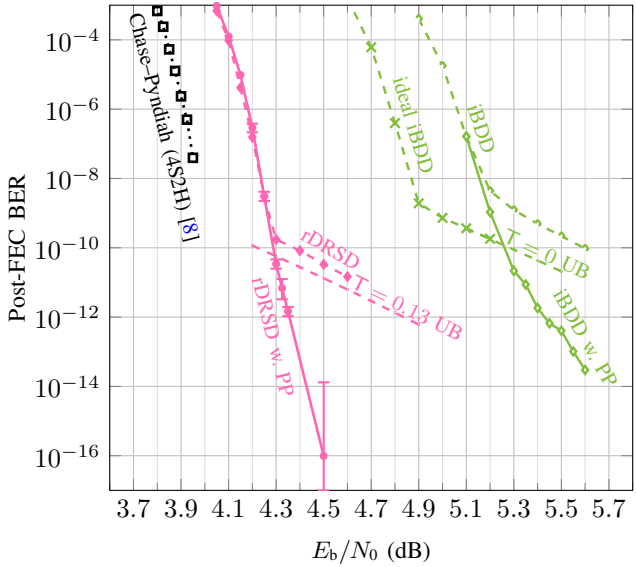}
        \vspace{-2ex}
    \end{tikzpicture}
    \caption{Post-FEC BER performance of the product code with \([256,239,6]\) eBCH component codes, comparing \ac{iBDD}, \acs{RDRSD}, and post-processing-based decoding schemes. The error bars indicate \(95\%\) confidence intervals.}
    \label{fig:BERcurve}
\end{figure}

We benchmark the proposed soft-aided post-processing (PP) technique against the method of~\cite{holzbaur2019improved}. Because post-processing operates in the error-floor region, where the post-decoding \ac{FER} is very low, conventional Monte-Carlo simulation is computationally demanding. We therefore propose a rejection-sampling-based simulation method.

For each simulated error pattern, we randomly select \(K\) rows and \(L\) columns and condition the \(K\times L\) intersection bits on hard-decision errors. For each transmitted bit \(x\), we draw \(r\sim\mathcal{N}(0,\sigma^2)\) and reject the sample unless
\(
    \psi\bigl(\mu(x)+r\bigr)\neq x .
\)
All remaining bits are generated according to the original channel model with independent noise \({r\sim\mathcal{N}(0,\sigma^2)}\). This inserts a \(K\times L\) hard-decision stall pattern while preserving the soft information distribution according to the channel condition.\footnote{A related evaluation approach was used in~\cite{holzbaur2019improved}, where a dedicated channel was constructed to deliberately insert certain stall patterns.
This approach is not directly suitable here, since the soft information associated with the received values affects the behavior of the proposed decoder and post-processing method.}
To gain insight into the post-processing techniques, we consider four dimensions \((K,L)\), namely \((3,3)\), \((3,4)\), \((3,5)\), and \((4,4)\), for the $[256,239,6]$ eBCH-code-based PC.
The corresponding results are shown in Fig.~\ref{fig:ISPP_all}.

With \acs{RDRSD}, a large fraction of the sampled stall patterns is already resolved during iterative decoding, since unreliable positions can be declared as erasures and subsequently recovered. In contrast, for \ac{iBDD} without post-processing, the post-decoding \ac{FER} remains close to \(1\) for all sampled patterns, as \ac{iBDD} cannot resolve such patterns unless rare miscorrections remove some of the errors.

Both post-processing techniques resolve small stall patterns with high probability. Across the considered cases, the proposed soft-aided post-processing technique outperforms the method of~\cite{holzbaur2019improved}. For the injected \((4,4)\) stall patterns, however, both methods show limited effectiveness, since these patterns can evolve into larger stall patterns if miscorrections turn component words with four errors into valid codewords.

In Fig.~\ref{fig:BERcurve}, post-processing is applied to both \ac{iBDD} and \acs{RDRSD}, and regular Monte Carlo simulations are performed to show the actual post-decoding \ac{BER} performance.
For \ac{iBDD}, we use the method of~\cite{holzbaur2019improved}, whereas for \acs{RDRSD}, we apply the proposed Algorithm~\ref{algo:PP}.
After post-processing, both decoders achieve significantly lower error rates than their counterparts without post-processing and also perform below the computed union bound (UB) for the miscorrection-free case.
For example, for \acs{RDRSD}, the union bound predicts an error floor around a \ac{BER} of \(10^{-10}\).
With post-processing, the steep decrease continues down to approximately \(10^{-12}\).
Below this level, the curve decreases more gradually due to the non-deterministic outcomes discussed in Sec.~\ref{subsec:largeStallPatterns}, but no error floor appears within the simulated range. At \(\Eb/\No=\SI{4.5}{dB}\), no errors occurred in \(1.5503\cdot 10^{11}\) simulated PC blocks. Assuming independent events of \(36\) erroneous bits per \(65536\)-bit block, the zero-count Poisson interval~\cite[Sec.~2.12]{van2011statistical} gives the one-sided \(95\%\) \ac{BER} upper bound \(1.06\cdot 10^{-14}\) shown in Fig.~\ref{fig:BERcurve}.

\section{Conclusion}\label{sec:conclusion}

In this paper, we presented the \acs{RDRSD} decoder for \acp{GPC}. Compared with conventional \ac{DRSD}, the proposed refinement reduces internal decoder data flow and memory requirements by maintaining a syndrome-domain implementation that is compatible with existing \ac{iBDD} architectures, while preserving a large decoding gain in the waterfall region. We also derived a union bound on the error floor of miscorrection-free error-and-erasure decoding and showed that the \acs{RDRSD} approaches this bound in the error-floor region. For PCs with component codes of small minimum distance, we proposed a soft-aided post-processing technique that further reduces the error floor. The combination of high decoding gain, reduced implementation overhead, and low error floors makes the proposed decoder a promising candidate for future low-complexity optical communication systems. An open-source implementation is provided at \url{https://github.com/kit-cel/DRSD4GPCs}.

\section*{Acknowledgement}
The authors would like to thank F.~Ritter, H.~Jäkel, S.~Obermüller, and L.~Rapp for the helpful discussion.

\ifCLASSOPTIONcaptionsoff
  \newpage
\fi

\bibliographystyle{IEEEtran}
\bibliography{bibtex/bib/IEEEabrv, bibtex/bib/lit, bibtex/bib/myown}

\end{document}